%% file: main.tex
\documentclass[lettersize,journal]{IEEEtran}

\usepackage{array}
\usepackage[font=normalsize,labelfont=sf,textfont=sf]{subcaption}
\usepackage{stfloats}
\usepackage{url}
\usepackage{verbatim}
\usepackage{enumitem}
\usepackage{bm}
\usepackage{multirow}
\usepackage{cite}
\usepackage{amsmath,amssymb,amsfonts,mathtools}
\usepackage{amsthm}
\usepackage{graphicx}
\usepackage{textcomp}
\usepackage{xcolor}
\newtheorem{theorem}{Theorem}
\newtheorem{proposition}{Proposition}

\usepackage[T1]{fontenc} % 设置字体编码为T1
\usepackage[utf8]{inputenc} % 设置输入编码为UTF-8
\usepackage{svg}
\usepackage{pifont}
\usepackage{booktabs}
\usepackage[table]{xcolor}
\usepackage{tabularx}

\usepackage{makecell}
\usepackage{algorithm}
\usepackage{algpseudocode}
\usepackage{tikz}
\usepackage{quantikz}

\begin{document}

\title{\LARGE UNIQ: Communication-Efficient Distributed Quantum Computing\\ via Unified Nonlinear Integer Programming}

\author{Hui~Zhong, Jiachen~Shen, Lei~Fan, Xinyue~Zhang, Hao~Wang, Miao~Pan, and Zhu~Han%
\thanks{Hui Zhong and Jiachen Shen contributed equally to this work.}%
\thanks{Hui Zhong, Jiachen Shen, Miao Pan, and Zhu Han are with the Department of Electrical and Computer Engineering, University of Houston, Houston, TX 77204 USA (e-mail: {hzhong5@uh.edu; jshen24@uh.edu; mpan2@uh.edu; hanzhu22@gmail.com}).}%
\thanks{Lei Fan is with the Department of Engineering Technology, University of Houston, Houston, TX 77204 USA, and also with the Department of Electrical and Computer Engineering, University of Houston, Houston, TX 77204 USA (e-mail: lfan8@central.uh.edu).}%
\thanks{Xinyue Zhang is with the Department of Computer Science, Kennesaw State University, Marietta, GA 30060 USA (e-mail: xzhang48@kennesaw.edu).}%
\thanks{Hao Wang is with the Department of Electrical and Computer Engineering, Stevens Institute of Technology, Hoboken, NJ 07030 USA (e-mail: hwang9@stevens.edu).}%
}

\markboth{IEEE JOURNAL ON SELECTED AREAS IN COMMUNICATIONS.}
{IEEE JOURNAL ON SELECTED AREAS IN COMMUNICATIONS.}


\maketitle

\begin{abstract}
Distributed quantum computing (DQC) is widely regarded as a promising approach to overcome quantum hardware limitations. A major challenge in DQC lies in reducing the communication cost introduced by remote CNOT gates, which are far slower and more resource-consuming than local operations. Existing DQC approaches optimize the three essential components (qubit allocation, entanglement management, and network scheduling) separately, or couple only a subset, even though these components are interdependent. We propose UNIQ, which integrates all three into a single non-linear integer programming (NIP) model: the qubit placement, the per-slot EPR inventory, and the gate schedule are coupled through one objective and one shared-resource constraint set, so a unified optimization approach can achieve lower-cost solutions. UNIQ reduces the circuit runtime by generating Einstein–Podolsky–Rosen (EPR) pairs in parallel on idle communication qubits, within the same slot as (or, within the coherence budget, just before) the gates that consume them, while minimizing the communication cost of remote gates. Because solving the NIP exactly does not scale, we solve it with a polynomial-time constructor: a greedy mapping of logical qubits to Quantum Processing Units (QPUs) followed by a Just-In-Time (JIT) schedule that builds EPR pairs in parallel within each slot, and we measure its optimality gap against the exact model. Simulations across diverse quantum circuits and QPU topologies show that UNIQ reduces communication cost and runtime over existing methods.
\end{abstract}


\section{Introduction\label{sec:introduction}}
\IEEEPARstart{D}{ue} to its unique quantum properties, quantum computing offers the potential to solve problems at an exponential speed compared to classical computing~\cite{bub2010quantum,jozsa2003role,rieffel2000introduction}. As a result, quantum computing is well-suited for solving complex problems~\cite{rietsche2022quantum,bova2021commercial,gyongyosi2019survey}, especially in molecular simulation~\cite{engkvist2000accurate}, drug discovery~\cite{cao2018potential}, and risk analysis~\cite{wilkens2023quantum}. To fully realize the advantages of quantum computing, practical quantum algorithms often require millions of qubits. However, current quantum hardware is still in the Noisy Intermediate-Scale Quantum (NISQ) era, where only a few hundred qubits are available~\cite{ju2024harnessing,zhao2024bridging,brooks2019beyond,zhong2024tuning,preskill2018quantum}. Moreover, due to fabrication challenges~\cite{brink2018device}, crosstalk errors~\cite{bruzewicz2019trapped} and quantum decoherence~\cite{brandt1999qubit}, scaling up the number of reliable qubits remains difficult in the short term, which restricts the advancement of quantum computing. To address this limitation, a widely accepted approach in both academia and industry is DQC~\cite{caleffi2024distributed,barral2025review,cacciapuoti2019quantum,beals2013efficient}. Similar to classical distributed computing, DQC enhances the computational power of quantum systems by interconnecting multiple QPUs via a quantum network, thereby increasing the total number of available qubits. For example, connecting two 128-qubit QPUs in a DQC system can functionally emulate a 256-qubit quantum computer.

The general workflow of DQC is typically divided into three main components: {\em qubit allocation, entanglement management, and network scheduling}~\cite{wu2022autocomm,wu2023qucomm,ferrari2023modular, liu2025co}. Qubit allocation divides large-scale quantum algorithms into multiple subcircuits and then assigns them to different QPUs for collaborative execution of the overall task. Gates within the same QPU are referred to as local gates, while those spanning across QPUs are referred to as remote gates. Since remote gates are much slower and more resource-consuming than local gates~\cite{davis2023towards}, a key optimization objective is to minimize the number of remote gates. Entanglement management is responsible for establishing EPR pairs and communication channels between QPUs to support the execution of remote gates. A commonly adopted approach is the Cat-Comm protocol, which relies on cat entanglement and disentanglement procedures~\cite{luo2024automatic}. This protocol uses specially designed circuits to entangle qubits across QPUs, while preserving the quantum information of the involved qubits. Network scheduling determines the execution order of remote gates by analyzing their dependency relationships. For example, if two remote CNOT gates share a common qubit, they must be executed in sequence. This is typically modeled as a directed acyclic graph (DAG)~\cite{chandra2024network}, where each node corresponds to a remote gate, and edges denote execution dependencies.

Several studies have explored the optimization of the DQC framework. Mao et al. (INFOCOM 2023)~\cite{mao2023qubit} focus on the qubit allocation problem for DQC and prove its NP-hardness. They propose an MHSA-based approach that combines local search and simulated annealing techniques, achieving better performance compared to existing methods. Similarly, Kan et al. (QCE 2024)~\cite{kan2024scalable} also concentrate on qubit allocation by merging sub-circuits. They dynamically adjust the partitioning based on the resource constraints of qubits, thereby reducing the number of circuit cuts. In contrast, Chandra et al. (TPS-ISA 2024)~\cite{chandra2024network} specialize in network scheduling for DQC, comparing a Resource-Constrained Project Scheduling Problem (RCPSP) framework with a greedy heuristic. Their results demonstrate that the two methods are suitable for different application scenarios, depending on circuit complexity. Further extending the scope, Zhou et al.~\cite{zhou2025cloudqc} shift the focus from a single-tenant to a multi-tenant DQC setting. They propose the CloudQC framework, which progressively optimizes qubit allocation and network scheduling to minimize the circuit runtime. More recently, two studies also pursue joint resource management for DQC. Liu et al.~\cite{liu2025co} co-design the network topology together with qubit allocation, treating the inter-QPU links as decision variables and embedding entanglement routing as flow constraints, then solve the resulting program directly with a classical solver. Bahrani et al.~\cite{bahrani2025resource} schedule a batch of whole circuits onto heterogeneous QPUs through a mixed-integer linear program, deciding circuit-to-QPU assignment and per-circuit partition counts, with entanglement generated on demand and consumed immediately. UNIQ differs from both on three points: it keeps the topology fixed and schedules individual remote gates into time slots; it maintains a per-slot EPR inventory and generates the pairs in parallel on idle communication qubits (in the execution slot by default, or earlier within the coherence budget), whereas these works generate entanglement on demand with no inventory; and it targets large single circuits through a polynomial-time heuristic, whereas both rely on exact solvers.

Recent work continues to advance these directions along three lines. In distributed quantum compilation, Promponas et al.~\cite{promponas2024compiler} use reinforcement learning to route EPR pairs and schedule operations across interconnected QPUs, Zhou et al.~\cite{zhou2025optimizing} compile circuits onto distributed systems through clustering and annealing, and Russo et al.~\cite{russo2025telesabre} extend layout synthesis to multi-core architectures with teleport interconnects. In quantum networking and entanglement routing, Abane et al.~\cite{abane2024entanglement} survey routing protocols that establish end-to-end entanglement through entanglement swapping, and Burt et al.~\cite{burt2025entanglement} partition circuits over large-scale networks in which entanglement cost grows with distance. These efforts advance compilation heuristics, network-layer routing, and circuit partitioning largely as separate concerns. UNIQ is complementary: it couples qubit allocation, entanglement management, and gate-level scheduling within one model and generates EPR pairs in parallel on idle communication qubits, and it can take the routing distances studied in this networking line as its cost coefficients $C_{u,v}$.

However, we found that most existing methods suffer from the following {\em limitations}. First, prior work typically divides DQC into the three components mentioned above and optimizes each of them separately. However, these components are inherently interdependent. Optimizing a single stage in isolation lacks a global view and often fails to achieve minimal circuit runtime. Second, since EPR generation time far exceeds gate execution time, the serial approach of establishing EPRs one by one before each remote gate leads to excessive latency. So, we need more efficient ways to reduce the total EPR generation time, for instance parallel EPR establishment or EPR reuse. Third, existing studies lack a systematic evaluation methodology. Some works compare their overall DQC frameworks with prior systems~\cite{wu2022autocomm,kan2024scalable}, while others only evaluate specific algorithms against known baselines (e.g., simulated annealing)~\cite{zhou2025cloudqc,mao2023qubit}. So the overall advantages of their DQC frameworks cannot be fully validated.

Therefore, we propose UNIQ (\textbf{U}nified \textbf{N}onlinear \textbf{I}nteger programming for Distributed \textbf{Q}uantum computing), a unified optimization approach for DQC that reduces the communication cost of remote gates and the circuit runtime. We first integrate qubit allocation, entanglement management, and network scheduling into one NIP model. Second, we divide the circuit runtime into uniform time slots and, within each slot, use idle communication qubits to establish EPR pairs in parallel. Third, we evaluate UNIQ with both algorithm-level comparisons against established methods and system-level comparisons against systems under similar settings, across diverse quantum circuits and DQC topologies. Our contributions can be summarized as follows:
\begin{itemize}
    \item We formulate DQC as a unified NIP model that couples qubit allocation, entanglement management, and network scheduling through a shared per-slot EPR-inventory constraint, optimizing the three stages against one objective and one feasible set, not as three separate problems.

    \item UNIQ generates the EPR pairs for concurrent remote gates in parallel on idle communication qubits, within their execution slot or, within the coherence budget, slightly earlier, which parallelizes entanglement generation and reduces the total EPR generation time.

    \item We solve the NIP with a polynomial-time constructor (greedy mapping followed by JIT scheduling) and run simulations on diverse quantum circuits and QPU topologies; UNIQ lowers circuit runtime over existing algorithms and systems at an acceptable execution time, and we measure how close the constructor is to the exact NIP optimum.

\end{itemize}

The remainder of this paper is organized as follows. Section~\ref{sec:preliminaries} introduces the background of quantum computing, DQC architectures, and quantum entanglement. Section~\ref{sec:system_model} presents the system abstraction and formulates the unified NIP optimization problem. Section~\ref{sec:framework} details UNIQ, including the greedy mapping and the JIT scheduling algorithm. Section~\ref{sec:evaluation} evaluates the effectiveness of our approach through extensive simulations. Finally, Section~\ref{sec:conclusion} concludes the paper.

\section{Preliminaries\label{sec:preliminaries}}
In this section, we first introduce the backgrounds of quantum computing~\cite{ju2024privacy,hey1999quantum,vedral1998basics,fan2022hybrid,hirche2023quantum}, and then outline the general workflow of DQC for solving large-scale quantum circuits.

\subsection{Qubits and Quantum Gates}
The basic unit of quantum computing is the qubit, which resides in a two-dimensional Hilbert space spanned by $\left | 0  \right \rangle$ and $\left | 1  \right \rangle$, similar to the 0 and 1 in classical computing. $\left | 0  \right \rangle$  and $\left | 1  \right \rangle$ can be represented as vectors $\left | 0  \right \rangle = (1,0)^{T}\ $ and $\left | 1  \right \rangle = (0,1)^{T}$. A quantum state describes the condition of a quantum system and can be classified as either a pure state or a mixed state. A pure state is a deterministic state. In a single-qubit system, a pure state can be either $\left | 0  \right \rangle$  or $\left | 1  \right \rangle$, or a superposition state $\left | \psi  \right \rangle = \alpha_{1}  \left | 0  \right \rangle + \alpha_{2} \left | 1 \right \rangle = (\alpha_{1},\alpha_{2} )^{T}\in \mathbb{C} ^{2} $, where $\alpha_{1}$ and $\alpha_{2}$ are complex numbers satisfying $\left | \alpha_{1}  \right | ^{2} + \left | \alpha_{2} \right | ^{2} = 1$. This means that a qubit can represent multiple values simultaneously. In an $n$ qubit system, a pure state can be represented as $\left | \psi  \right \rangle = \sum_{i=1}^{2^{n}}\alpha _{i}\left | i  \right \rangle  \in \mathbb{C} ^{2^{n}}$, where the coefficients satisfy $\left | \alpha_{1}  \right | ^{2} +...+ \left | \alpha_{2^{n}} \right | ^{2} = 1$. When the system interacts with the environment or experiences noise, it may evolve into a mixed state. A mixed state is described by a density matrix of the form $\rho =\sum_{i=1}^{2^{n}} p_{i} \left | \psi  \right \rangle \left\langle\psi\right| $, where $n$ is the number of qubits. Each $\left|\psi_i\right\rangle$ is a pure state and $p_{i}$ denotes the probability of the system being in pure state $\left|\psi_i\right\rangle$ with $\sum_{i=1}^{2^{n}} p_{i} =1$.

Quantum gates are fundamental operations on qubits, analogous to logical gates in classical computing. Mathematically, each quantum gate is represented by a unitary matrix $U$, which satisfies $U^{\dagger} U = U U^{\dagger} = I$. $U^{\dagger}$ denotes the conjugate transpose of $U$, and $I$ is the identity matrix. These gates govern the evolution of quantum states, and different sequences of quantum gates implement different quantum algorithms. Specifically, when the quantum state passes through a quantum gate $U$, the original states will become a new quantum state $\rho'=U\rho U^{\dagger}$. Moreover, any multi-qubit quantum gate can be decomposed into a series of basic quantum gates (for example single-qubit gates and CNOT gates). Common single-qubit gates include the X gate ($X\left | 0  \right \rangle =\left | 1  \right \rangle $, $X\left | 1  \right \rangle =\left | 0  \right \rangle $), analogous to the classical NOT gate. Common two-qubit gates include the CNOT gate, which flips the target qubit if the control qubit is $\left|1\right\rangle$.

%; the Y gate ($Y\left | 0  \right \rangle =i\left | 1  \right \rangle $, $Y\left | 1  \right \rangle =-\left | 0  \right \rangle $), which combines bit and phase flipping; the Z gate ($Z\left | 0  \right \rangle =\left | 0  \right \rangle $, $Z\left | 1  \right \rangle =-\left | 1  \right \rangle $), which flips the phase of the $\left|1\right\rangle$ component; and the H gate ($H\left | 0 \right \rangle =\frac{\left | 0 \right \rangle+\left | 1 \right \rangle }{\sqrt{2} } $,  $H\left | 1 \right \rangle =\frac{\left | 0 \right \rangle-\left | 1 \right \rangle }{\sqrt{2} } $), which transforms a basis state into a superposition.

\subsection{Distributed Quantum Computing}

\begin{figure}[t]
\centering
\includegraphics[width=0.45\textwidth]{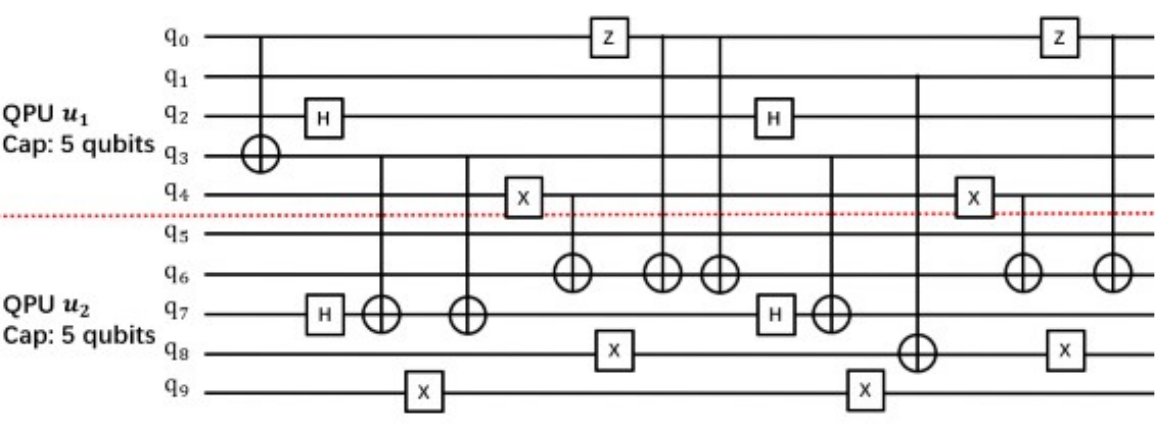}
\caption{Solving a large-scale problem via DQC. A quantum circuit is partitioned into two subcircuits, each assigned to a different QPU for parallel execution.}\label{fig:circuit_cut}
\vspace{-3mm}
\end{figure}

\begin{figure}[t]
\centering
\includegraphics[width=0.45\textwidth]{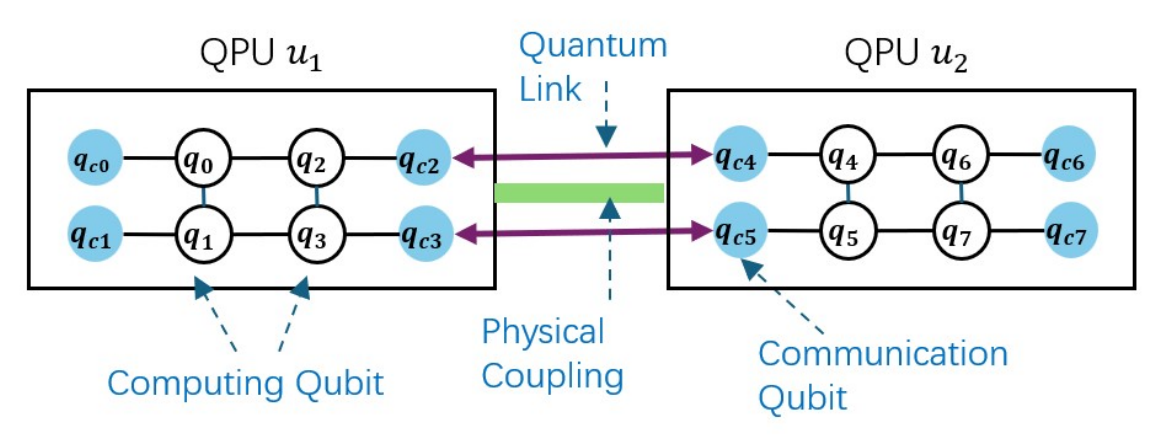}
\caption{DQC architecture with two QPUs.}\label{fig:dqc_archi}
\vspace{-3mm}
\end{figure}

\begin{figure}[t]
\centering
\includegraphics[width=0.35\textwidth]{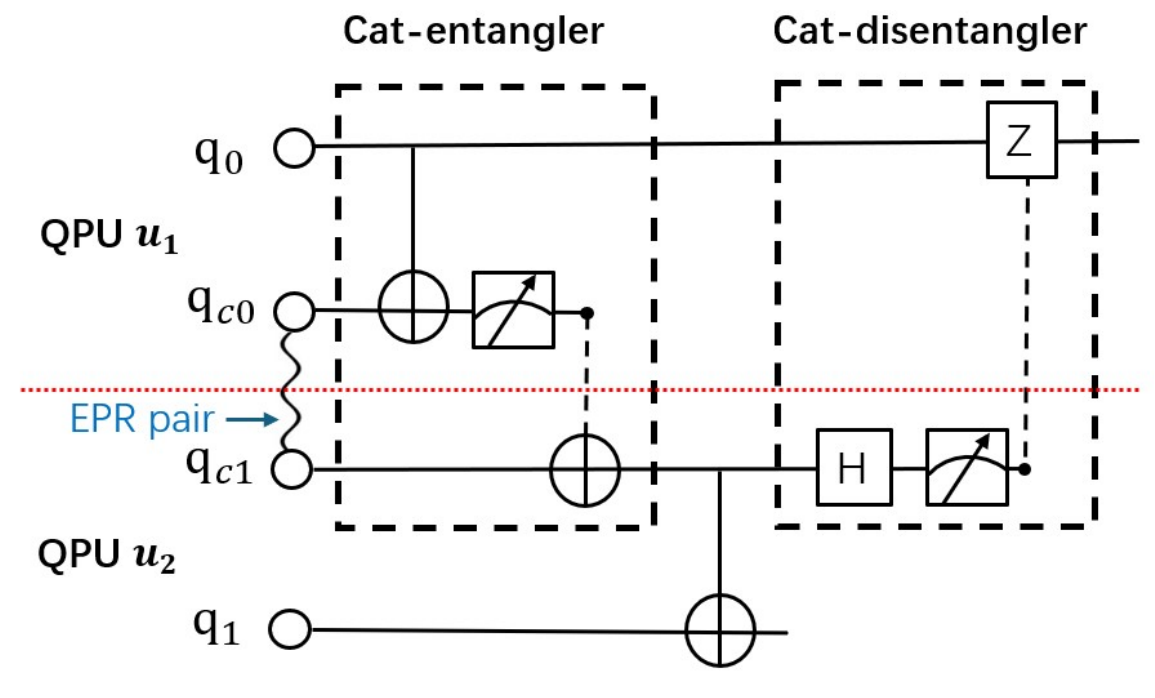}
\caption{Cat-Comm implementation of one remote CNOT gate. $q_0$ and $q_1$ are computing qubits; $q_{c0}$ and $q_{c1}$ are communication qubits. $q_0$ and $q_{c0}$ belong to QPU $u_1$; $q_1$ and $q_{c1}$ belong to QPU $u_2$.}\label{fig:cat_comm}
\vspace{-3mm}
\end{figure}

\input{epr_time}

DQC integrates multiple QPUs through network infrastructure to increase the number of available qubits, thereby enabling large-scale quantum computing tasks. Take Fig.~\ref{fig:circuit_cut} as an example: assuming that the task requires 10 qubits, while each QPU can only accommodate 5 qubits. Therefore, this circuit must be distributed across at least two QPUs. Figure~\ref{fig:dqc_archi} shows a DQC architecture example consisting of two QPUs. Each QPU contains two types of qubits: computing qubits (white circle) for executing local gates and communication qubits (blue circles) for establishing channels between QPUs, also referred to as quantum links. When two QPUs are physically connected (physical coupling), their communication qubits can establish quantum links, enabling the execution of remote CNOT gates. In addition, although the numbers of computational and communication qubits per QPU in DQC are not identical, they are assumed to be approximately equal. The two qubit types are bounded by separate constraints because they limit different resources: the number of computing qubits caps how many logical qubits reside on a QPU, whereas the number of communication qubits caps how many EPR pairs the QPU can hold at the same time. A QPU that maintains a single pool of $K_u$ qubits for both purposes is the special case $Cap_u+E_u=K_u$, where the split between the two roles becomes a design choice. The DQC framework is primarily divided into three components: qubit allocation, entanglement management, and network scheduling.

Qubit allocation involves splitting a quantum circuit and assigning subcircuits to different QPUs in a rational manner. The primary objective is to minimize the number of remote CNOT gates, which occur across QPUs and incur higher communication costs compared to local gates. As an example, shown in Fig.~\ref{fig:circuit_cut}, the first CNOT gate operates on two qubits within the same QPU and is referred to as a local CNOT gate; while the second CNOT gate spans two QPUs and is referred to as a remote CNOT gate. Entanglement management is a key technique for implementing remote CNOT gates. Currently, the most widely used method is Cat-Comm, as illustrated in Fig.~\ref{fig:cat_comm}~\cite{mao2023qubit}. QPU $u_1$ and QPU $u_2$ first generate an EPR pair by establishing a quantum link between their respective communication qubits. As shown in Table~\ref{tab:epr_latency}, CX is two-qubit local gate, and the EPR preparation is approximately 12 times that of a CX gate, which is the most time-consuming step in the entire process~\cite{zhou2025cloudqc, wu2022autocomm}. These latencies are representative values adopted from prior DQC studies~\cite{zhou2025cloudqc, wu2022autocomm} and are normalized to the CX-gate time, since the absolute entanglement-generation latency, and its ratio to a local gate, depend on the hardware platform (for example, superconducting, trapped-ion, or photonic systems). The reported ratio reflects the order-of-magnitude gap between heralded EPR establishment and a local two-qubit gate on current hardware. Our model uses the slot length $t=t_{ep}$ as its time unit and depends only on which operations fit within a slot, so substituting a platform's own ratio rescales the slots without changing the formulation or the algorithm. The EPR pair is then combined with local gates, measurements, and classical communication to implement a remote CNOT gate between qubits $q_0$ and $q_1$. EPR pairs are thus generated by the QPUs themselves, locally at the communication qubits at each end of a physical link, with no central source. When the two QPUs of a remote CNOT are directly connected, this single elementary pair suffices. When they are not, an elementary pair is generated on each link of the shortest path between them, and the intermediate QPUs fuse these pairs into a single end-to-end pair by entanglement swapping; the number of hops, $C_{u,v}$, is then the number of elementary pairs and swaps consumed, which is the communication cost charged per remote gate in the objective. The per-slot communication-qubit budget $E_u$ is enforced at the two endpoint QPUs, which hold the end-to-end pair until it is consumed. Communication qubits on the intermediate relay QPUs that perform entanglement swapping are abstracted into the hop count $C_{u,v}$ and are not separately budgeted. Another commonly used method is TP-Comm~\cite{wu2022autocomm}. However, this approach requires measuring the qubit $q_0$, which collapses its quantum state. It alters the original circuit structure and will introduce additional retransmission latency. Therefore, we only adopt Cat-Comm to implement remote gates in this paper. Network scheduling refers to the execution order of quantum gates based on their dependencies. Specifically, any two gates that share at least one qubit must follow a sequential execution order. For example, in Fig.~\ref{fig:circuit_cut}, the third CNOT gate can only be applied after the second CNOT has been completed. In contrast, gates can be executed in parallel when quantum resources are available. DAG is commonly used to represent the dependency and execution constraints in quantum circuits. For more details, refer to~\cite{zhou2025cloudqc}.

\section{System Modeling and Problem Formulation}\label{sec:system_model}
\input{notation}
In this section, we present the modeling and formulation for UNIQ. We first abstract the DQC environment into a time-slotted model to capture the discrete nature of quantum operations and communication events. Based on this abstraction, we then rigorously formulate the joint optimization of qubit allocation, entanglement management, and network scheduling as a NIP problem, aiming to minimize circuit runtime and communication costs under physical resource constraints.

\subsection{System Abstraction and Assumptions}
\label{subsec:model_abstraction}

\subsubsection{Time Slot Modeling and Gate Simplification.} We divide the total execution time of the quantum circuit into multiple time slots of length $t$, where $t$ corresponds to the EPR pair establishment time ($t_{ep}$ in Table~\ref{tab:epr_latency}). Since the execution time of a remote gate is much longer than that of a local gate, multiple local gates can be executed within a single time slot. Therefore, the execution time of local gates is negligible and can be ignored in our model. We retain only local CNOT gates and remote CNOT gates in the subsequent circuit representation, as illustrated in Fig.~\ref{fig:ex_circuit}.

\subsubsection{Unified Optimization of DQC Stages.} Most of the previous work dealt with the three stages of DQC separately. Each stage is typically optimized with its own objective function. Such design flow lacks global coordination and often results in suboptimal solutions. Indeed, we observe that the three stages of DQC are inherently interdependent. Circuit allocation affects the number and location of remote gates, while these remote gates determine the need and timing of entanglement generation. The above two steps directly constrain the feasibility and efficiency of network scheduling. Therefore, we propose a unified modeling approach to achieve more coordinated solutions.

\subsubsection{EPR Pre-establishment for Remote Gate Efficiency.} Although prior research has recognized that remote CNOT gates are time-consuming and attempted to reduce this (e.g., using the same EPR pair for consecutive remote gates~\cite{ferrari2023modular}), these methods face some limitations. They often require that the involved gates reside on the same pair of QPUs and disallow any intermediate operations. To overcome these constraints, we propose an alternative approach: pre-establishing EPR pairs for future remote gates. When there are still idle communication qubits available in a time slot, additional EPR pairs can be created in advance. These entangled pairs can then be directly used by subsequent remote operations. By establishing multiple EPR pairs in parallel within a single time slot, the overall circuit runtime can be reduced.

\subsubsection{Quantum Nature of the Managed Resource.} The techniques used in the solver (greedy placement, an EPR inventory, and earliest-feasible scheduling) are standard in form, but the resource they manage is entanglement, whose non-classical properties shape the model. First, EPR pairs cannot be copied (no-cloning), so a pre-established pair cannot be duplicated to serve several gates; each remote gate is bound to exactly one pair (Constraint~e), and the inventory is therefore not a shareable cache. Second, executing a remote gate consumes the pair through the Cat-Comm measurement, so a pair is single-use and the inventory is decremented at consumption (Constraint~i). Third, a stored pair decoheres: its fidelity decays with the waiting time $\Delta$ as $F(\Delta)=\tfrac12+(F_0-\tfrac12)e^{-\Delta t_{ep}/T_{coh}}$, so its validity is time-bounded and a pre-established pair must be consumed within the coherence budget (Constraint~j), unlike a classical pre-allocated resource, which does not expire. Fourth, a pair is established non-locally by heralded entanglement across a physical link, and for non-adjacent QPUs by entanglement routing and swapping, so the communication cost $C_{u,v}$ counts entanglement-swap hops. These constraints, together with the routing model, are what make the formulation quantum; the classical-form algorithm operates within them.

\subsubsection{Design Objectives of UNIQ-DQC.} We employ a unified objective function to jointly optimize the following objectives: 
\begin{itemize}
    \item Minimize the communication cost of remote gates: The QPUs involved in remote gates may not be adjacent, so the system must identify the nearest path to execute the remote gate efficiently.

    \item Minimize total task runtime under acceptable algorithm execution time: The circuit runtime should be reduced as much as possible, while ensuring that the algorithm's execution time remains within an acceptable cost.
\end{itemize}

\subsection{Mathematical Formulation}
\label{subsec:problem_formulation}

In this section, we formulate the DQC problem based on the system model described above. The goal is to execute a large-scale quantum task over a DQC network while satisfying the design objectives outlined earlier. To achieve this, we jointly consider the three interdependent stages of DQC (qubit allocation, entanglement management, and network scheduling) to coordinate quantum gate operations across the entire circuit. The parameters, decision variables, and auxiliary variables used in our formulation are summarized in Table~\ref{tab:parameters}. First, we introduce our objective function:

\begin{equation}
\min \alpha\sum_{t=1}^{H} \sum_{g \in \mathcal{G}} t \cdot z_{g,t}
+ \beta \sum_{t=1}^{H} \sum_{\substack{u,v \in \mathcal{U} \\ u \ne v}} \sum_{g=(i_g, j_g)} C_{u,v} \cdot \pi_{i_g, u} \cdot \pi_{j_g, v} \cdot z_{g,t}.
\label{eq:objective}
\end{equation}
The objective function consists of two parts: minimizing the weighted flow time (surrogate for makespan) and reducing the communication cost of remote CNOT gates. The first component of Eq.~(\ref{eq:objective}) adopts a compactness-oriented objective by minimizing the sum of completion times. Unlike directly minimizing the makespan ($C_{\max}$), which may allow intermediate gates to be loosely scheduled without improving the final completion time, this formulation forces all gates to be executed as early as possible. This produces a dense schedule that reduces the total circuit runtime (makespan) while crucially reducing the idle time of qubits in memory, thereby mitigating decoherence errors; Proposition~\ref{prop:sched} states the condition under which this schedule is provably optimal. 

The second component focuses on reducing the communication cost incurred by remote CNOT gates. $C_{uv}$ is defined as the shortest path length between QPU $u$ and  QPU $v$ as~\cite{zhou2025cloudqc}. If $\pi_{i_g, u} \cdot \pi_{j_g, v}=1$, it indicates that the two qubits involved in the remote gate $g$ are located on QPUs $u$ and $v$; otherwise, they do not. $\alpha$ and $\beta$ are two weighting parameters to balance these two components.

Next, we will introduce ten constraints of our model. For clarity, each constraint is illustrated with a corresponding circuit example, as shown in Figs.~\ref{fig:ex_circuit} and~\ref{fig:ex_circuit2}.

\paragraph{Mapping Validity}
Each logical qubit $q$ must be assigned to exactly one QPU $u$. As shown in Fig.~\ref{fig:ex_circuit}, $q_0$ and $q_1$ are mapped to QPU $u1$; $q_2$ and $q_3$ are mapped to QPU $u2$.
\begin{equation}
\sum_{u \in \mathcal{U}} \pi_{q,u} = 1, \quad \forall q \in Q.
\label{c1}
\end{equation}

\paragraph{QPU Capacity}
The number of allocated qubits $q$ on a QPU $u$ cannot exceed its capacity. As shown in Fig.~\ref{fig:ex_circuit}, at most two qubits can be assigned to QPU $u1$ or $u2$.
\begin{equation}
\sum_{q \in Q} \pi_{q,u} \leq Cap_{u}, \quad \forall u \in \mathcal{U}.
\end{equation}

\paragraph{Gate Scheduling}
Each gate $g$ must be and can only be assigned an exact time slot $t$. As shown in Fig.~\ref{fig:ex_circuit}, gate $g_1$ is executed only at time slot $t_1$; gate $g_2$ is executed only at time slot $t_2$.
\begin{equation}
\sum_{t=1}^{H} z_{g,t} = 1, \quad \forall g \in \mathcal{G}.
\end{equation}

\paragraph{Remote-Gate Indicator}
A binary variable $\delta_g$ indicates whether the two qubits of CNOT gate $g = (i_g, j_g)$ reside on different QPUs. If $\delta_g=0$, it means that the two qubits of gate $g$ are on the same QPU (local CNOT gate); if $\delta_g=1$, it means that the two qubits of gate $g$ are on different QPUs (remote CNOT gate). As shown in Fig.~\ref{fig:ex_circuit}, $g_3$ is local CNOT and $\delta_g=0$; $g_4$ is remote CNOT and $\delta_g=1$.
\begin{equation}
\delta_g = 1 - \sum_{u \in \mathcal{U}} \pi_{i_g,u} \pi_{j_g,u}, \quad \forall g \in \mathcal{G}.
\label{eq:placement}
\end{equation}

\paragraph{EPR Generation Requirement}
If gate $g$ is a remote CNOT gate ($\delta_g=1$), then one EPR pair for that gate is required. This constraint explicitly couples the computing and communication resources: the execution of a remote gate on computing qubits ($\delta_g=1$) strictly demands the consumption of an entangled pair generated by communication qubits ($y_{g,t}=1$), reflecting the physical necessity of entanglement consumption for teleporting quantum states in the Cat-Comm protocol. As shown in Fig.~\ref{fig:ex_circuit}, $g_3$ does not require EPR and $\sum_{t=1}^{H} y_{g_{3},t}=0$; $g_4$ requires EPR and $\sum_{t=1}^{H} y_{g_{4},t}=1$.
\begin{equation}
\sum_{t=1}^{H} y_{g,t} = \delta_g, \quad \forall g \in \mathcal{G}.
\end{equation}
\begin{figure}[t]
\centering
\includegraphics[width=0.45\textwidth]{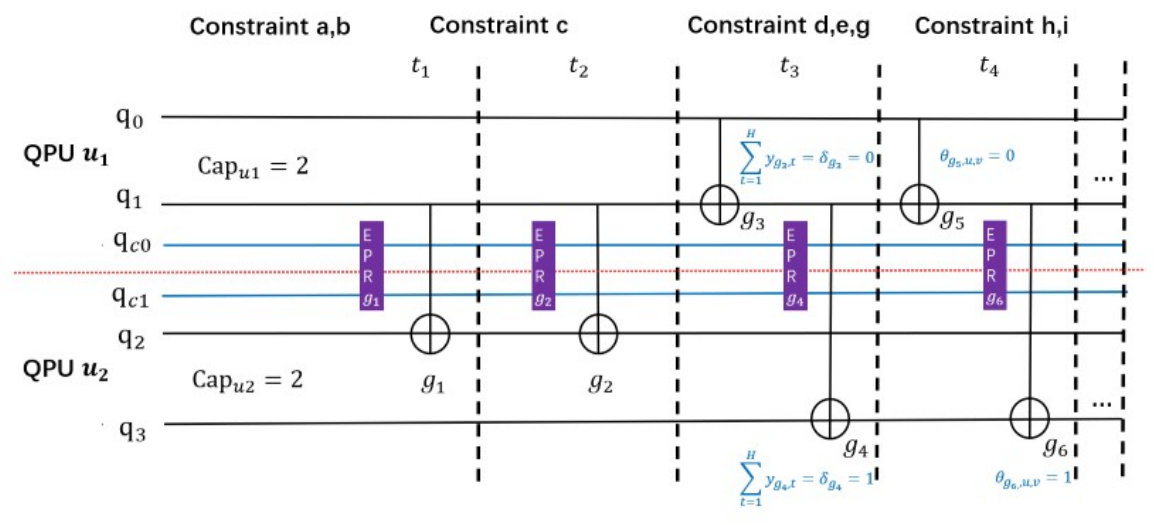}
\caption{Example circuit illustrating constraints a-e, g-i.}\label{fig:ex_circuit}
\vspace{-3mm}
\end{figure}

\begin{figure}[t]
\centering
\includegraphics[width=0.4\textwidth]{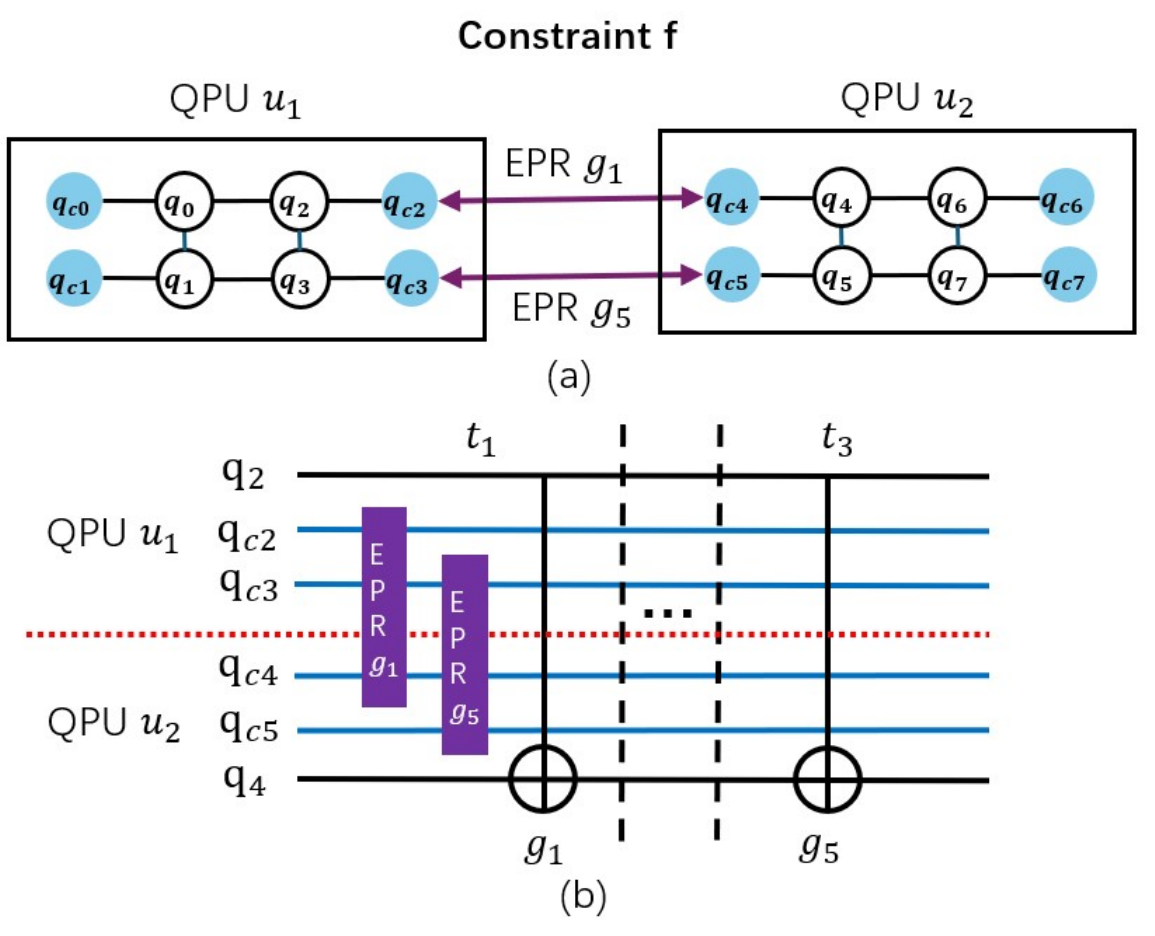}
\caption{(a) Two EPR pairs for $g_1$ and $g_5$ are generated in parallel and completed together in the same time slot. (b) Example circuit illustrating constraint f.}\label{fig:ex_circuit2}
\vspace{-3mm}
\end{figure}
\paragraph{EPR Before Execution Ordering}
The EPR pair for the remote gate $g$ must be generated before the gate is executed. As shown in Fig.~\ref{fig:ex_circuit2}, $g_5$ is executed at $t_3$, and thus its required EPR pair must be established no later than $t_3$. Suppose that communication qubits $q_{c3}$ and $q_{c5}$ are idle at $t_1$; in this case, the EPR pair for $g_5$ can be pre-established at $t_1$. Moreover, since two EPR pairs for $g_1$ and $g_5$ are generated in parallel within the same time slot, only one EPR setup slot is consumed instead of two, thereby reducing the circuit runtime.
\begin{equation}
\sum_{\tau=1}^{t} y_{g,\tau} \geq z_{g,t} - (1 - \delta_g), \quad \forall g \in \mathcal{G},\ \forall t=1,\dots,H.
\end{equation}

\paragraph{Precedence}
If two CNOT gates have qubit overlap, then they have execution order. As shown in Fig.~\ref{fig:ex_circuit}, $g_4$ must be scheduled no later than $g_6$.
\begin{equation}
\sum_{\tau=1}^{t} z_{g',\tau} \leq \sum_{\tau=1}^{t} z_{g,\tau}, \quad \forall (g', g) \in \mathcal{P},\ \forall t=1,\dots,H.
\end{equation}

\paragraph{Undirected Mapping Indicator}
A dummy variable $\theta_{g,u,v}$ indicates whether gate $g$ spans QPUs $u$ and $v$. As shown in Fig.~\ref{fig:ex_circuit}, $g_5$ is local CNOT, and $\theta_{g_{5},u,v}=0$; $g_6$ is remote CNOT, and $\theta_{g_{6},u,v}=1$.
\begin{equation}
\theta_{g,u,v} = \pi_{i_g,u} \pi_{j_g,v} + \pi_{i_g,v} \pi_{j_g,u}, \quad \forall g \in \mathcal{G},\ \forall u, v \in \mathcal{U},\ u \neq v.
\end{equation}

\paragraph{Concurrent EPR Generation and Inventory}
EPR pairs can be stored and used, subject to inventory limits. As shown in Fig.~\ref{fig:ex_circuit}, the EPR pair generated at $t_4$ occupies a communication qubit on each endpoint during $t_4$, so $s_{u,v,t_{4}}=1$; the remote gate consumes it in the same slot, which frees the qubit for later slots ($s_{u,v,t_5}=0$).
\begin{equation}
s_{u,v,0} = 0, \quad \forall u,v \in \mathcal{U},
\end{equation}
\begin{equation}
\begin{aligned}
&s_{u,v,t} = s_{u,v,t-1}
+ \sum_{g \in \mathcal{G}} y_{g,t} \theta_{g,u,v} \\
&\quad \quad - \sum_{g \in \mathcal{G}} z_{g,t-1} \theta_{g,u,v},
\quad \forall u \neq v,\ t = 1,\dots,H,
\end{aligned}
\end{equation}
\begin{equation}
0 \leq \sum_{v \in \mathcal{U}} s_{u,v,t} \leq E_u, \quad \forall u \in \mathcal{U},\ \forall t=1,\dots,H.
\label{c11}
\end{equation}
The consumption term is offset by one slot, so a pair occupies its endpoints from its generation slot through its execution slot and is released only afterward; an in-slot pair therefore counts against $E_u$ during its execution slot, and Eq.~\eqref{c11} bounds the number of pairs held simultaneously at each QPU.

\paragraph{Coherence Budget}
A pre-established EPR pair loses fidelity while it waits, so it must be consumed within its coherence budget after generation. Each communication qubit has a finite coherence time $T_{coh}$, and a Bell pair held for $\Delta$ slots between generation and consumption has fidelity $F(\Delta)=\tfrac12+(F_0-\tfrac12)e^{-\Delta t_{ep}/T_{coh}}$, where $F_0$ is the fidelity at generation. We set the coherence budget $L_{coh}=\lfloor (T_{coh}/t_{ep})\ln\frac{F_0-1/2}{F_{th}-1/2}\rfloor$ to the largest storage time that keeps fidelity at or above a target $F_{th}$. For every remote gate, the EPR pair must be consumed no later than $L_{coh}$ slots after it is generated. As shown in Fig.~\ref{fig:ex_circuit2}, if $g_5$ is executed at $t_3$, then its EPR pair cannot be generated earlier than slot $t_3-L_{coh}$.
\begin{equation}
\sum_{\tau=1}^{t} y_{g,\tau} \leq \sum_{\tau'=1}^{\,t+L_{coh}} z_{g,\tau'}, \quad \forall g \in \mathcal{G},\ \forall t=1,\dots,H.
\label{c12}
\end{equation}
Together with the EPR-before-execution ordering, this bounds the storage time as $0 \leq \tau(g)-t_{\mathrm{gen}}(g) \leq L_{coh}$.

\section{Framework Design\label{sec:framework}}

% MODIFIED: Added Section Introduction (Requirement 2)
In this section, we present the algorithmic realization of the UNIQ model. To tackle the computational complexity of the NIP formulation, we design a constructive heuristic that decomposes the problem into two sequential stages: spatial qubit mapping and temporal gate scheduling. This structure allows us to efficiently navigate the trade-off between global optimality and computational scalability, ensuring practical applicability for large-scale quantum circuits.

\subsection{Design Overview}
\label{subsec:design-overview}

Now we introduce our UNIQ-DQC framework. The inputs to the UNIQ include the logical qubit set \(Q\), the CNOT multiset
\(\mathcal G\) with partial order \(\mathcal P\), a cloud
configuration \((\mathcal U,\mathrm{Cap},E,C)\), and a slot
horizon \(H\leq|\mathcal G|\). 

While the NIP model formulated in Section~\ref{sec:system_model} provides a global optimization perspective, directly solving it is computationally intractable for large-scale circuits due to its NP-hard nature. Therefore, we adopt a two-stage heuristic approach (Greedy Mapping followed by Just-In-Time (JIT) Scheduling) to achieve a scalable and efficient solution.
We adopt a Greedy–JIT constructor to generate a full execution plan that maps the quantum circuit onto the distributed quantum processor. This plan ensures that all hardware constraints are satisfied, including qubit capacity limits, precedence, and communication requirements between QPUs. It is constructed according to the rules defined by the hybrid QAP–RCPSP (Quadratic Assignment Problem-Resource Constrained Project Scheduling Problem) model. What is new here is not the use of quadratic assignment or project scheduling individually, but their coupling through a shared physical resource: the per-slot EPR inventory that the placement loads and the scheduler consumes, together with the EPR-before-execution ordering and the coherence budget. This coupling is what makes the three DQC stages a single joint problem, not a placement problem followed by an independent scheduling problem, and the optimality-gap study in Section~\ref{sec:evaluation} quantifies how little is lost by solving it in stages.

Our primary objective is to minimize the circuit runtime and communication cost as defined in Eq.~(1). Additionally, considering the scalability requirement, we design the solver to operate within low polynomial time complexity.

The UNIQ workflow is structured as a two-stage pipeline consisting of qubit allocation and network scheduling. In the first stage, logical qubits are greedily mapped to physical QPUs based on communication weights. This mapping is determined once at the beginning and stored in the \(\pi\), which remains fixed throughout the process. In the second stage, each gate is scheduled in the earliest available time slot that satisfies communication constraints, based on the DAG precedence and the fixed placement \(\pi\). If a gate spans two QPUs, the required EPR pair is generated at the same time slot or earlier, thereby avoiding long-term reservation of communication qubits.

Separating qubit allocation from network scheduling narrows the search space relative to solving the NIP directly, and this decomposition is what keeps the solver scalable to large circuits. The greedy allocation prioritizes immediate interaction cost and does not look ahead to bursty remote-gate demand; the JIT scheduling phase recovers this by smoothing peak demand within the fixed mapping. The decomposition trades global temporal search for polynomial-time solvability. The two stages remain coupled through the model: the placement determines which gates are remote and at what communication cost, and feeds the same objective and per-slot EPR-inventory constraints that the scheduler must respect, so the decomposition solves the joint model in sequence, not as two independent objectives. Section~\ref{sec:evaluation} measures the optimality gap that this decomposition introduces against the exact NIP.

\subsection{Greedy Qubit--QPU Mapping\label{subsec:mapping}}

We first fix a qubit--to--QPU assignment, which remains unchanged during the subsequent scheduling stage. This placement phase is feasibility-driven: each logical qubit must be mapped to exactly one QPU and no QPU exceed its qubit capacity (Constraints a-b). Among all feasible placements, we prioritize those that are likely to reduce remote communication costs during execution. 

\begin{figure}[t]
  \centering
\includegraphics[width=0.5\textwidth]{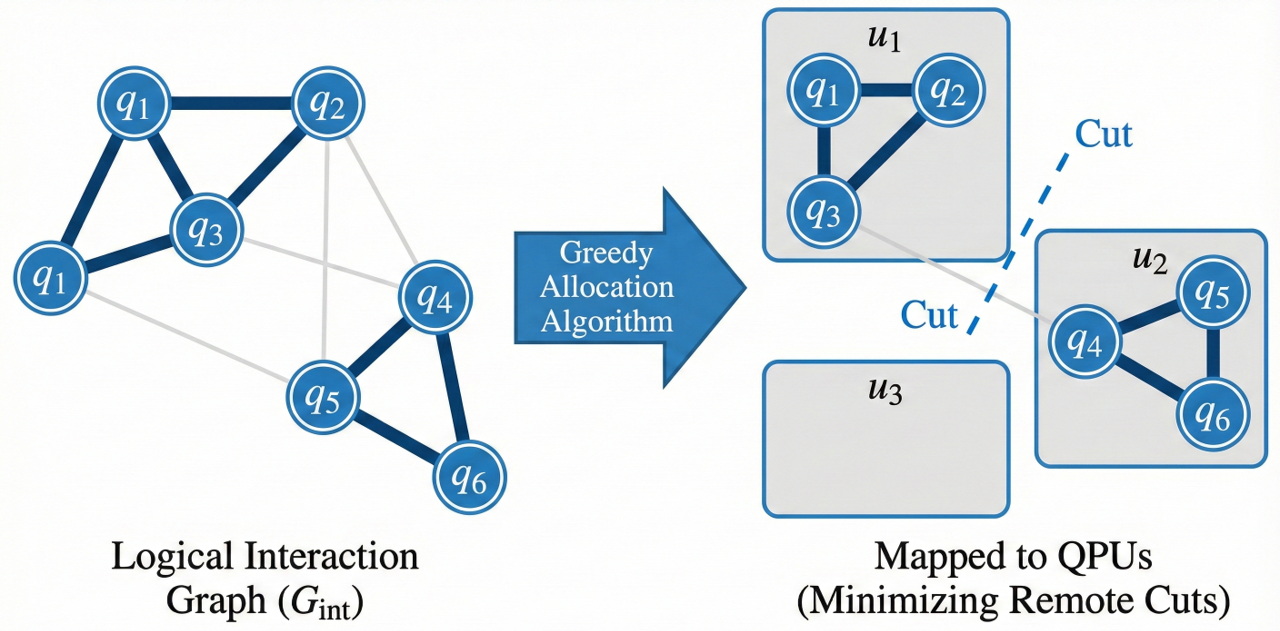}
  \caption{Illustration of the Greedy Qubit-QPU Mapping. Logical qubits with high interaction weights (thick edges) are grouped into the same QPU to minimize the cut weight of remote connections.}
  \label{fig:mapping_viz}
\end{figure}

To guide this process, we construct an \emph{interaction graph} \(G_{\mathrm{int}}=(Q,w)\), whose vertices are the logical qubits and whose edge weights quantify two-qubit gate interactions. Specifically, the weight $w_{ij} = \bigl|\{\,g\in\mathcal G \mid \{i_g,j_g\}=\{i,j\}\,\}\bigr|$, $w_{ii}=0$, counts the number of CNOT gates between qubits \(i\) and \(j\).
Heuristically, assigning vertices with large \(w_{ij}\) to the same QPU reduces the number of remote CNOTs. As illustrated in Fig.~\ref{fig:mapping_viz}, the algorithm clusters qubits connected by heavy edges (solid lines) into the same physical QPU (e.g., QPU 1), ensuring that most operations remain local. Consequently, only the sparse, low-weight connections (dashed lines) cross the QPU boundaries, forming a minimal cut that lowers the communication cost and EPR pair consumption.

The procedure performs a single deterministic sweep over all qubits. For each qubit $q$, we first compute its total interaction weight
\(W(q)=\sum_{j\in Q} w_{qj}\) and process the qubits in non-increasing order of \(W\) (ties are broken by the qubit index so that the algorithm is reproducible). At each step, only QPUs with remaining capacity are considered. Let
\(S(q)=\{u\in\mathcal U \mid r_u>0\}\) denote the set of QPUs with remaining capacity, where \(r_u\) counts free seats on device \(u\).
For each QPU \(u\in S(q)\), we evaluate the new communication cost incurred by assigning qubit \(q\) to \(u\):
$\Delta(q,u)
\;=\;
\sum_{j:\ \pi(j)\ \mathrm{assigned}} w_{qj}\, C_{u,\pi(j)}$,
where \(\pi(j)\) denotes the QPU already assigned to qubit \(j\) and \(C_{u,\pi(j)}\) is the inter-QPU distance (shortest-path hops, with \(C_{u,u}=0\)). The sum runs over all already-placed qubits: each neighbour \(j\) contributes its interaction weight \(w_{qj}\) scaled by the distance from \(u\) to \(j\)'s QPU, which is zero when \(j\) is on \(u\) (a local interaction) and grows with the distance otherwise. The qubit is then assigned to the QPU \(u^\star\) that minimizes \(\Delta(q,u)\) (once again the ties are resolved by the smallest QPU index). Minimizing \(\Delta\) places \(q\) near the QPUs holding its heavy neighbours, so the placement accounts for both the number of remote interactions and the inter-QPU distance, not the gate count alone. Per-link EPR availability and quality are not part of the placement objective; they are enforced downstream by the scheduler's per-slot communication-qubit budget (Constraint i) and the coherence budget (Constraint j), so the placement uses the inter-QPU distance \(C_{u,v}\) as the cost proxy and the scheduler guarantees feasibility at execution time. The assignment \(\pi_{q,u^\star}=1\) is recorded and the residual capacity \(r_{u^\star}\) decreased.
Since a qubit is placed only when \(r_{u^\star}>0\), the uniqueness and capacity invariants are preserved throughout, and no backtracking is required.

Once all qubits have been mapped, the placement matrix \(\pi\) is used to derive the gate indicators required by the model. For each gate \(g=(i_g,j_g)\),
 $\delta_g \;=\; 1 - \sum_{u\in\mathcal U} \pi_{i_g,u}\pi_{j_g,u}$
marks whether the gate is local (\(\delta_g=0\)) or remote (\(\delta_g=1\)). For each ordered pair \(u\neq v\),
 $\theta_{g,u,v}
 \;=\;
 \pi_{i_g,u}\pi_{j_g,v} + \pi_{i_g,v}\pi_{j_g,u}$
identifies the (unordered) QPU endpoints of a remote gate. These arrays, along with \(\pi\), are passed unchanged to the scheduling phase.

Two edge cases are handled explicitly. First, if a qubit never participates in a CNOT (\(W(q)=0\)), all candidate devices give \(\Delta(q,u)=0\), and the deterministic tie-breaker chooses the first device with remaining capacity. Second, if at any point, \(S(q)=\varnothing\) (that is, the total capacity \(\sum_u \mathrm{Cap}_u\) is insufficient), the procedure ends and reports the infeasibility of the instance under Constraints a-b. This entire procedure is summarized in Algorithm~\ref{alg:mapping}.

\begin{algorithm}[t]
\caption{Greedy Qubit--QPU Mapping}
\label{alg:mapping}
\small
\begin{algorithmic}[1]
\Require Interaction weights $w_{ij}$, QPU capacities $\mathrm{Cap}_u$, communication costs $C_{uv}$
\Ensure Mapping $\pi_{q,u}$ with $\sum_{u}\pi_{q,u}=1$ and $\sum_{q}\pi_{q,u}\le \mathrm{Cap}_u$
\State $\pi_{q,u}\gets 0$ for all $q,u$; \quad $r_u\gets \mathrm{Cap}_u$ for all $u$
\State $W(q)\gets \sum_{j\in Q} w_{qj}$ for all $q$
\State $\mathcal{O}\gets$ qubits sorted by non--increasing $W(q)$, breaking ties by index
\For{$q \in \mathcal{O}$}
 \State $S(q)\gets \{\,u\in\mathcal U \mid r_u>0\,\}$
\If{$S(q)=\varnothing$}
 \State \textbf{return} \textsc{Infeasible} \Comment{$\sum_u \mathrm{Cap}_u < |Q|$}
\EndIf
\State $\displaystyle u^\star \gets \arg\min_{u\in S(q)} \ \sum_{j:\ \pi(j)\ \mathrm{assigned}} w_{qj}\, C_{u,\pi(j)}$
\Comment{interaction $\times$ distance; ties by smallest $u$}
\State $\pi_{q,u^\star}\gets 1;\quad r_{u^\star}\gets r_{u^\star}-1$
\EndFor
\end{algorithmic}
\end{algorithm}
% \vspace{-4mm}

\subsection{JIT Scheduling with EPR Generation}
\label{subsec:scheduling}

Once the mapping matrix \(\pi\) (and therefore the indicators \(\delta,\theta\) in Constraints d and h) is fixed, each CNOT gate \(g\in\mathcal G\) has to be assigned in a time slot \(\tau(g)\in\{1,\ldots,H\}\) while obeying precedence, single-slot assignment, and the per-slot communication-qubit (EPR) budgets (Constraints c, e-g, i). The scheduler operates on the precedence DAG \((\mathcal G,\mathcal P)\) under a fixed topological order. For each gate \(g\), we compute the earliest slot that can possibly host it,
$t_{\min}(g) = 1 + \max\{\,\tau(g') \mid (g',g)\in\mathcal P\,\}$,
with the maximum over the empty set is \(0\) by convention. This choice directly encodes the precedence relation into a lower bound on the start time, so any slot \(t\ge t_{\min}(g)\) automatically respects Constraint g.

Starting from \(t_{\min}(g)\), the scheduler linearly scans the timeline and selects the first slot which satisfies all resource constraints. This earliest-feasible rule promotes a compact schedule by greedily minimizing the surrogate objective \(\sum_t t\,z_{g,t}\) and avoids delaying gate execution unnecessarily toward the end of the horizon. Once a feasible slot is found, the gate is assigned to it, which ensures Constraint c.

When $\delta_g=0$ (both qubits of $g$ reside on the same QPU), the first tested slot is accepted because no inter-QPU communication resource is consumed. In this case, the EPR-related constraints (Constraints e-f, i) are vacuous, whereas precedence (Constraint g) and single-slot assignment (Constraint c) are already satisfied by the earlier construction. When $\delta_g=1$ (remote CNOT), the scheduler additionally verifies that both endpoint QPUs have sufficient communication capacity at the candidate slot $t$:
$\sum_{v\neq u} s_{u,v,t} \;<\; E_u, \text{for } u=\pi(i_g),\ \pi(j_g)$,
where $s_{u,v,t}$ is the EPR inventory maintained by the recursion (Constraint i). This test counts the EPR pairs already reserved on each endpoint at slot $t$; the gate is allowed to reserve an additional pair. Therefore, the inequality ensures the per-slot budget in Constraint i. When the budget is already full at the candidate slot, the check fails and the scan advances to the next slot, so remote-gate demand that exceeds the per-slot communication capacity is deferred and spread across later slots, with the makespan absorbing the overflow. The instance is reported infeasible only if no slot within the horizon has free capacity on both endpoints, in particular when a QPU has no communication qubit.

Once a feasible execution slot $t^\star$ is found, we set $z_{g,t^\star}=1$ and record $\tau(g)=t^\star$. For a remote gate, we then choose a generation time
$t_{\mathrm{gen}}(g)\in\{1,\ldots,\tau(g)\}$,
and set $y_{g,t_{\mathrm{gen}}(g)}=1$. The default choice is \emph{in-slot} generation $t_{\mathrm{gen}}(g)=\tau(g)$, which minimizes the EPR lifetime and thus inventory pressure. Alternatively, we may take an earlier slot in $[\max(1,\tau(g)-L_{coh}),\,\tau(g)-1]$ to smooth peak demand; the lower limit enforces the coherence budget of Constraint~(j) so that a pre-established pair is consumed before its fidelity falls below $F_{th}$. Since in-slot generation is always available and incurs no storage decoherence ($\Delta=0$), while pre-establishing the same pair earlier occupies a communication qubit for longer without shortening the makespan, in-slot is the default and pre-establishment is invoked only to relieve congestion within the coherence budget.
\begin{figure}[t]
  \centering
  \includegraphics[width=0.5\textwidth]{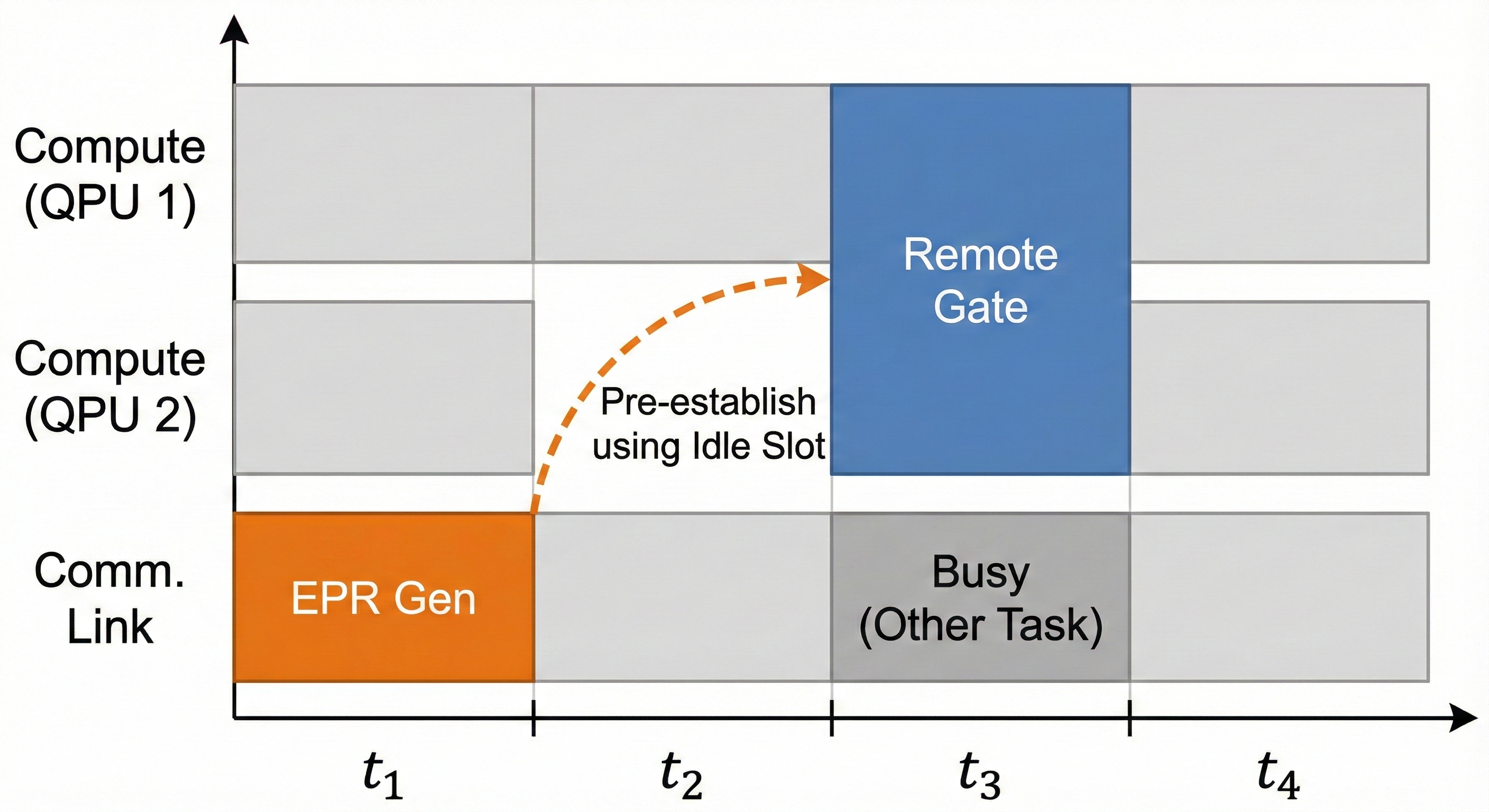}
  \caption{Visualizing the JIT Scheduling mechanism. An EPR pair required for a remote gate at $t_3$ is proactively generated in the idle communication slot at $t_1$, decoupling resource generation from consumption.}
  \label{fig:jit_viz}
\end{figure}
As visualized in Fig.~\ref{fig:jit_viz}, this mechanism allows the scheduler to decouple the resource generation from consumption. Specifically, if the communication link at the execution slot $t_3$ is congested, UNIQ can proactively instruct the idle communication qubits at $t_1$ to generate the required EPR pair (orange block) in advance. In both cases, the same per-slot budget check is enforced at $t_{\mathrm{gen}}(g)$ for both endpoint QPUs. This satisfies the requirement of generating exactly one EPR pair if and only if the gate is remote (Constraint e) and the ordering requirement that the pair be available no later than execution (Constraint f).

After committing $(z,y)$, the inventory is updated to keep it consistent with Constraint~i: the pair occupies a communication qubit on each endpoint from its generation slot $t_{\mathrm{gen}}(g)$ through its execution slot $\tau(g)$, and it is removed (consumed) afterward, so it occupies no later slot. This consumption update is essential: it frees the communication qubit for subsequent gates and keeps $\sum_v s_{u,v,t}\le E_u$ accurate at every slot. For in-slot generation the pair occupies only the execution slot. Maintaining $s$ in this way guarantees consistency for future capacity checks and ensures that the equalities and bounds in Constraint i hold inductively. Generation is demand-driven: each remote gate generates exactly one EPR pair (Constraint e) and consumes it, by default in the same slot, while local gates generate none. UNIQ therefore performs no speculative pre-sharing, so every generated pair is used and none decoheres unused, giving an EPR utilization of one.

By identifying and using idle communication qubits to pre-establish EPR pairs, UNIQ sustains throughput even when communication resources are scarce.

Since the scan for each gate only moves forward in time, we are guaranteed to find a feasible slot unless the instance itself is infeasible under the given capacities. Specifically, if the scan exceeds the horizon \(H\) without finding a valid slot, the algorithm terminates and reports the instance as infeasible under the current horizon.
Throughout, the algorithm maintains: (i) the fixed completion times \(\tau(\cdot)\) for computing new \(t_{\min}(\cdot)\); (ii) the binary matrices \(z\) and \(y\); and (iii) \(s\) (or an equivalent per-QPU per-slot aggregation from which \(s\) can be reconstructed). Since no decision is ever revisited, the routine terminates after a single pass over \(\mathcal G\). The complete procedure is summarized in Algorithm~\ref{alg:scheduler}.

\begin{algorithm}[t]
\caption{JIT schedule with EPR generation}
\label{alg:scheduler}
\small
\begin{algorithmic}[1]
\Require fixed mapping $\pi$, indicators $\delta$; precedence DAG $(\mathcal G,\mathcal P)$; horizon $H$
\State $z_{g,t}\gets 0,\; y_{g,t}\gets 0 \quad \forall g,t$
\State $s_{u,v,t}\gets 0 \quad \forall u\neq v,\; t=1,\dots,H$
\For{$g$ in topological order of $(\mathcal G,\mathcal P)$}
\State $t \gets 1 + \max_{g'\prec g}\tau(g')$ \Comment{0 if no predecessor}
 \While{$t \le H$}
 \If{$\delta_g = 0$}
 \State $z_{g,t}\gets 1;\;\tau(g)\gets t$ \Comment{local gate, no cross-QPU resource}
\State \textbf{break}
\Else
\State $u_1 \gets \pi(i_g),\;\; u_2 \gets \pi(j_g)$
\State $\text{used}_{u_1} \gets \sum_{v\neq u_1} s_{u_1,v,t}$,\quad
$\text{used}_{u_2} \gets \sum_{v\neq u_2} s_{u_2,v,t}$
 \If{$\text{used}_{u_1} + 1 \le E_{u_1}$ \textbf{and} $\text{used}_{u_2} + 1 \le E_{u_2}$}
\State $z_{g,t}\gets 1;\;\tau(g)\gets t$
\State $t_{\mathrm{gen}}\gets t$ \Comment{in-slot default; or a slot in $[\max(1,t-L_{coh}),t]$ to pre-establish}
\State $y_{g,t_{\mathrm{gen}}}\gets 1$
\State $s_{u_1,u_2,\tau'}\mathrel{+}{=}1,\; s_{u_2,u_1,\tau'}\mathrel{+}{=}1\;\; \forall\,\tau'\in[t_{\mathrm{gen}},t]$ \Comment{generation: pair occupies its endpoints from $t_{\mathrm{gen}}$ through execution}
\State consume the pair at $\tau(g)=t$: it is released and occupies no slot $\tau'>t$ \Comment{inventory update at execution (added)}
\State \textbf{break}
 \Else
\State $t \gets t+1$
\EndIf
\EndIf
\EndWhile
\textbf{\If{$t > H$}
\State \textbf{return} \textsc{Infeasible}
\EndIf}
\EndFor
\end{algorithmic}
\end{algorithm}

This minimal-state design has several practical advantages. First, each slot test only accesses two 32-bit integers, making the inner loop primarily memory-bound. As a result, the procedure scales efficiently with the number of hardware threads available. Second, an EPR pair for a remote gate may be generated either in the same slot as the gate’s execution or in an earlier one. Since the EPR pair is always consumed no later than the gate’s execution, its lifetime is finite, and the number of reserved EPRs remains bounded over time. Third, by avoiding speculative EPR storage, the algorithm reduces the probability of failure on unreliable quantum links. This advantage is supported by empirical evidence from recent prototypes, which report EPR success rates around 30\%~\cite{van2007communication}. Finally, the deterministic property of the procedure means that any observed performance variation is caused by the hardware rather than the scheduling algorithm. This is particularly valuable for cloud providers aiming to deliver consistent latency within service-level agreements.

\subsection{Theoretical Analysis}
\label{subsec:analysis}

The following two theorems establish two key properties of the procedure: (i) it is guaranteed to terminate after a polynomial number of basic computational steps; (ii) it invariably returns a schedule that satisfies the full set of constraints defined in Section~\ref{sec:system_model}. Overall, these results confirm that the Greedy--JIT constructor provides a reliable and predictable fast start for any subsequent optimization stage.

\begin{theorem}[Convergence]\label{thm:termination}
For every finite instance
\(
(Q,\mathcal G,\mathcal P,\mathcal U,
\mathrm{Cap},E,C,H)
\), the Greedy--JIT constructor terminates after at most
$O\!\bigl(
n\log n + np + m+e + mH + p^{2}H
\bigr)
\;=\;
O\!\bigl(m^{2}\bigr)$
primitive operations when
\(n,e,H=\Theta(m)\) and \(p=o(m)\).
\end{theorem}

\begin{proof}
The placement routine of Section \ref{subsec:mapping} begins by sorting all $n$ logical qubits, which takes $O(n\log n)$ time, followed by scanning at most \(p\) QPUs for each qubit, contributing an additional \(np\) operations.
Next, a topological sort of the precedence graph (representing gate dependencies) is performed, which visits
all \(m\) vertices and all \(e\) arcs and therefore costs \(m+e\) operations.
In the scheduling phase (Section \ref{subsec:scheduling}), the algorithm inspects at most \(H\) candidate slots for each of the \(m\) gates. After all gates are fixed, it updates a \(p\times p\times (H+1)\) inventory array, which takes \(p^{2}H\) steps.
Summing these contributions yields the bound stated in the theorem. Under the scaling \(n,e,H=\Theta(m)\) and \(p=o(m)\), the leading term is \(m^{2}\).
\end{proof}

\begin{theorem}[Space Complexity]\label{thm:space}
The Greedy--JIT constructor requires $O(n + e + p^2 H)$ space complexity to store the circuit structure, mapping states, and EPR inventory.
\end{theorem}
\begin{proof}
The interaction graph $G_{\mathrm{int}}$ and placement map $\pi$ require $O(n^2)$ (dense case) or $O(n+e)$ (sparse case) and $O(n)$ space, respectively. The precedence DAG necessitates $O(m+e)$ storage. The critical component is the EPR inventory tensor $s_{u,v,t}$, which consumes $O(p^2 H)$ space. Since $m, e \gg n$, the total space complexity is dominated by $O(e + p^2 H)$, ensuring memory efficiency for typical NISQ-era parameters.
\end{proof}

\begin{theorem}[Feasibility]\label{thm:feasibility}
Let \((\pi,z,y,\delta,\theta)\) be the schedule returned by the constructor. Then, under the condition \(T_{coh}\ge t_{ep}\), all ten constraints a-j from section \ref{sec:system_model} are satisfied.
\end{theorem}

\begin{proof}
\textit{a--b (mapping).}
A qubit is mapped only once and only to a QPU with residual capacity, hence \(\sum_{u}\pi_{q,u}=1\) and
\(\sum_{q}\pi_{q,u}\le\mathrm{Cap}_{u}\).

\textit{c (gate uniqueness).}
In Algorithm~\ref{alg:scheduler}, a gate \(g\) is committed to the first
feasible slot and the loop breaks immediately (local case:
Lines~7--8; remote case: Lines~13--17). Thus exactly one \(z_{g,t}\) is set to~1 and \(\sum_{t=1}^H z_{g,t}=1\).

\textit{d (same--QPU indicator).}
After placement \(\delta_g\) is computed as in Eq.~(\ref{eq:placement}), and both cases
\(\delta_g\in\{0,1\}\) are obtained directly from \(\pi\).

\textit{e-f (EPR requirements).}
If \(\delta_g=0\), no EPR is generated (\(\sum_t y_{g,t}=0\));
if \(\delta_g=1\), the scheduler sets \(y_{g,\tau(g)}=1\) and no other entry, so \(\sum_ty_{g,t}=1\) and
\(\sum_{\tau\le t}y_{g,\tau}\ge z_{g,t}-\bigl(1-\delta_g\bigr)\).

\textit{g (precedence).}
Because each candidate slot \(t\) is tested in non--decreasing order starting from
\(1+\max_{g'\prec g}\tau(g')\), we have \(\tau(g)>\tau(g')\) for all \(g'\prec g\), and therefore
\(\sum_{\tau\le t}z_{g',\tau}\le\sum_{\tau\le t}z_{g,\tau}\).

\textit{h (unordered QPU pair).}
After mapping, \(\theta_{g,u,v}\) is set exactly as required by Constraint h.

\textit{i (EPR inventory).}
The inventory is initialized at Algorithm~\ref{alg:scheduler} Line~2 with \(s_{u,v,0}=0\) and updated at
Line~16 is consistent with the recursion in Constraint i. Before accepting a remote gate, the algorithm checks (Algorithm~\ref{alg:scheduler} Lines~12--13) that
\(\sum_{v\neq u}s_{u,v,t}+1 \le E_u\) for both endpoints, which enforces the per--slot budget. Hence \(0\le\sum_v s_{u,v,t}\le E_u\) for all \(u,t\), and the update equations in Constraint i hold.

\textit{j (coherence budget).}
For a remote gate, the default in-slot generation sets \(t_{\mathrm{gen}}(g)=\tau(g)\), so the storage time \(\Delta_g=0\le L_{coh}\) whenever \(T_{coh}\ge t_{ep}\); when pre-establishment is used, the generation slot lies in \([\max(1,\tau(g)-L_{coh}),\tau(g)-1]\), so \(\Delta_g\le L_{coh}\) by construction. Because in-slot generation is always available, a feasible schedule exists under the condition \(T_{coh}\ge t_{ep}\), and Constraint j holds.

Therefore, each constraint condition has been satisfied.
\end{proof}

Beyond feasibility, the scheduling stage admits a structural optimality guarantee in the regime where the per-slot communication budget is not the bottleneck. For a fixed placement $\pi$, let $d(g)$ denote the length of the longest precedence chain ending at gate $g$, i.e. $d(g)=1+\max_{g'\prec g} d(g')$ with $d(g)=1$ when $g$ has no predecessor. Any schedule that respects precedence must place $g$ no earlier than slot $d(g)$.

\begin{proposition}[Scheduler optimality under a non-binding budget]\label{prop:sched}
Fix a placement $\pi$. If, for every slot $t$ and every QPU $u$, the number of remote gates $g$ with $d(g)=t$ and $u\in\{\pi(i_g),\pi(j_g)\}$ is at most $E_u$, then the JIT earliest-feasible schedule assigns $\tau(g)=d(g)$ for every gate, and this simultaneously minimizes the makespan $\max_g \tau(g)$ and the flow-time term $\sum_g \tau(g)$ of Eq.~(\ref{eq:objective}) over all precedence-respecting schedules.
\end{proposition}

\begin{proof}
We argue by induction on the topological order. The scheduler places $g$ in the first slot $t\ge 1+\max_{g'\prec g}\tau(g')$ whose per-slot budget is not exceeded at both endpoints. Assume $\tau(g')=d(g')$ for all $g'\prec g$; then $1+\max_{g'\prec g}\tau(g')=d(g)$. By hypothesis, at slot $d(g)$ the number of remote gates sharing an endpoint QPU with $g$ and due at $d(g)$ is at most $E_u$, so the budget admits $g$ and it is placed at $\tau(g)=d(g)$. Since $d(g)$ is a lower bound on $\tau(g)$ for every precedence-respecting schedule, attaining $\tau(g)=d(g)$ for all $g$ minimizes each $\tau(g)$ individually, hence both $\max_g\tau(g)$ and $\sum_g\tau(g)$.
\end{proof}

Proposition~\ref{prop:sched} explains the small optimality gap observed for the scheduling stage in Section~\ref{sec:evaluation}: whenever communication capacity meets the peak per-slot demand, the schedule is optimal for the runtime term given the placement, so the residual gap to the global optimum is concentrated in the placement. We do not, however, give a worst-case approximation guarantee for the joint problem, and this is a limitation of the problem, not of the algorithm: the placement subproblem is a quadratic assignment problem, which admits no constant-factor polynomial-time approximation unless P${=}$NP~\cite{sahni1976pcomplete}. We therefore substantiate the advantage of the unified treatment empirically, through the optimality-gap and ablation studies in Section~\ref{sec:evaluation}, and not through a worst-case bound.

These results jointly imply that the Greedy--JIT constructor is a deterministic, polynomial-time algorithm whose output is guaranteed to be feasible for the hybrid QAP--RCPSP model. Consequently, it provides a provably reliable and computationally inexpensive baseline, which can be used as a warm start for exact solvers or meta-heuristic refinement procedures.

\section{Evaluation\label{sec:evaluation}}
In this section, we present an evaluation of UNIQ. We first present a representative example to intuitively demonstrate the effectiveness of UNIQ. Next, we compare UNIQ with existing algorithms across a range of quantum circuits and QPU topologies, and benchmark it against existing DQC frameworks. We then measure the optimality gap against the exact NIP and the exact solver's scalability, analyze EPR coherence and fidelity under the coherence budget, and test robustness to local-gate latency and heterogeneous QPU resources.

\begin{figure}[t]
\centering
\includegraphics[width=0.45\textwidth]{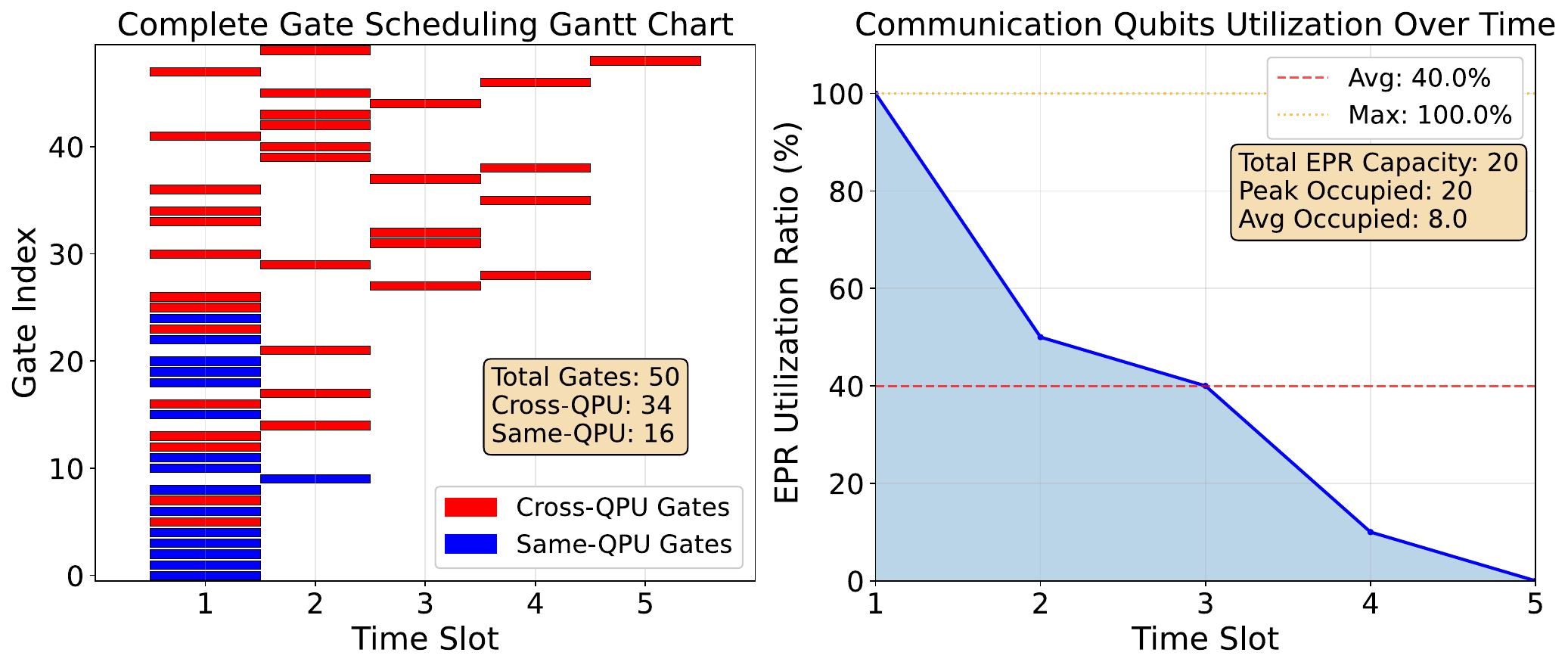}
\caption{Performance of representative example.}\label{fig:ep_example}
\vspace{-3mm}
\end{figure}

\input{user_circuit} % Commented out for independent compilation

\begin{figure}[t]
\centering
\includegraphics[width=0.45\textwidth]{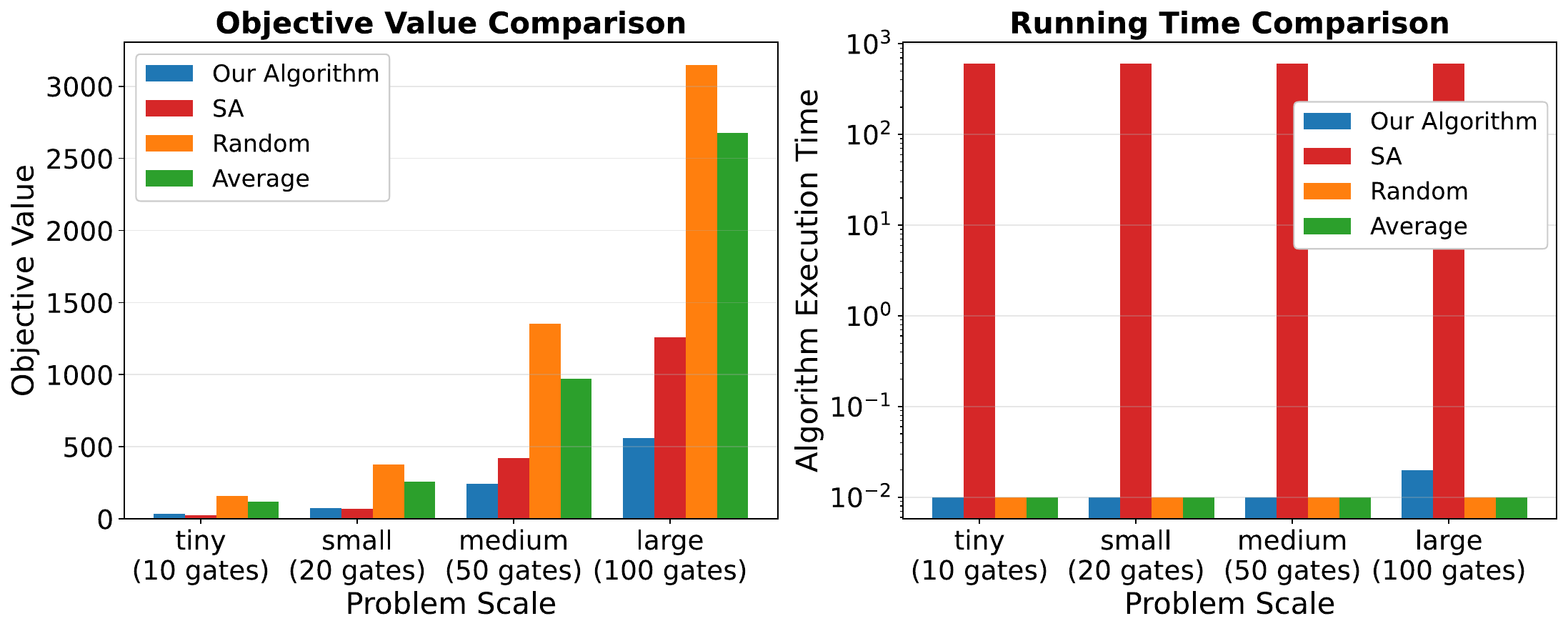}
\caption{Performance on different user-defined circuits.}\label{fig:ep_user_cir}
\vspace{-3mm}
\end{figure}

 \input{circuit_info} % Commented out for independent compilation

\begin{figure*}[t]
% \centering
\begin{subfigure}{0.45\textwidth}
    \centering
    \includegraphics[width=\textwidth]{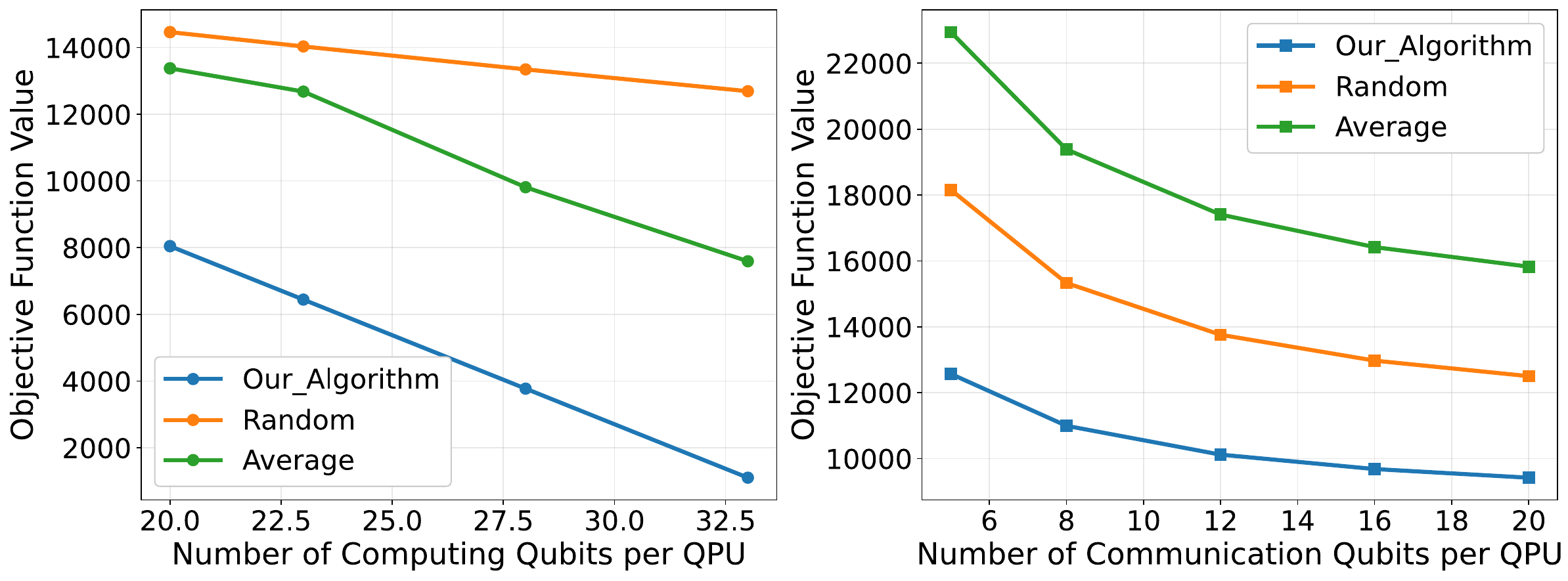}
    \caption{bv\_n70 circuit}
    \label{fig:bv}
\end{subfigure}
\hfill
\begin{subfigure}{0.45\textwidth}
    \centering
    \includegraphics[width=\textwidth]{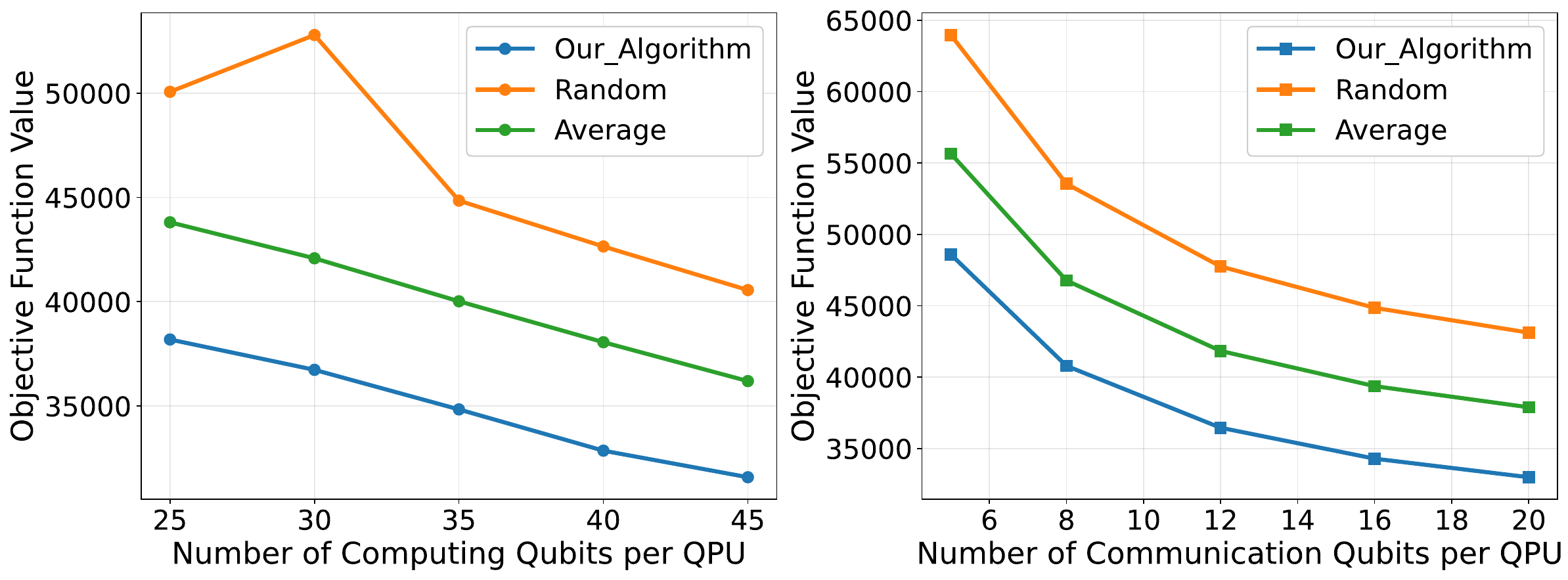}
    \caption{cat\_n130 circuit}
    \label{fig:cat}
\end{subfigure}

\vspace{2mm}

% 第二行：GHZ 和 QuGAN
\begin{subfigure}[b]{0.45\textwidth}
    \centering
    \includegraphics[width=\textwidth]{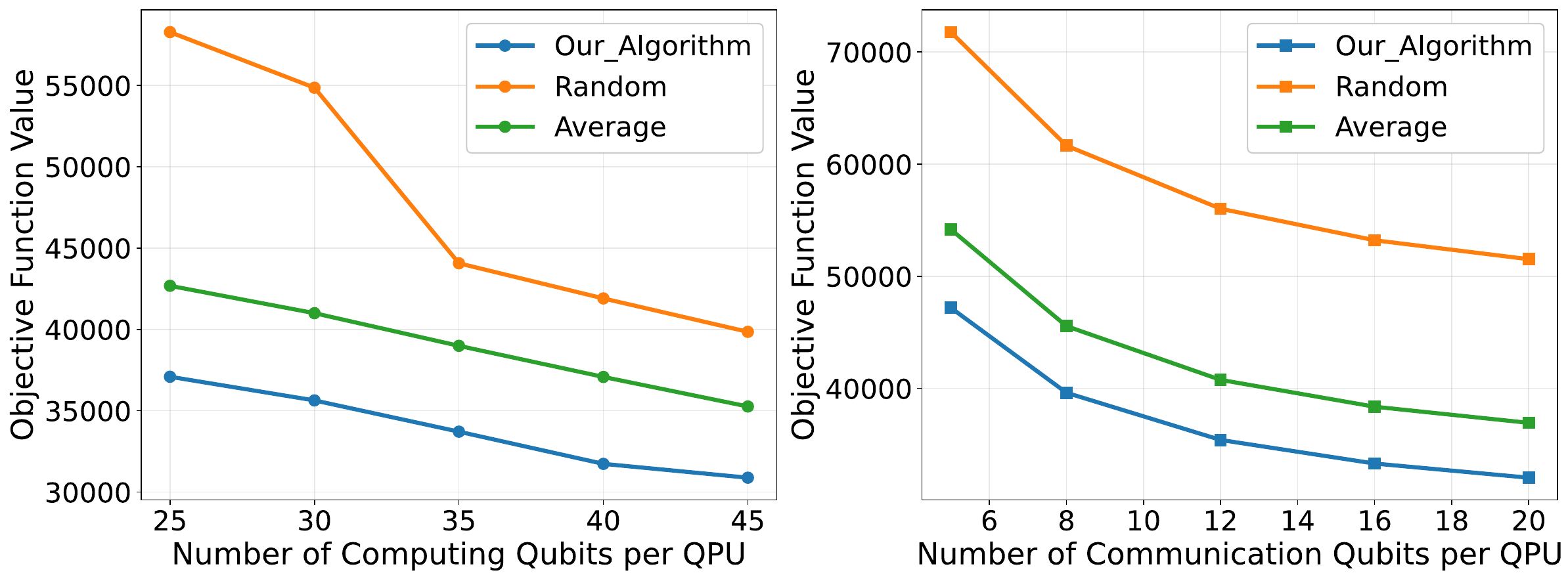}
    \caption{ghz\_n127 circuit}
    \label{fig:ghz}
\end{subfigure}
\hfill
\begin{subfigure}[b]{0.45\textwidth}
    \centering
    \includegraphics[width=\textwidth]{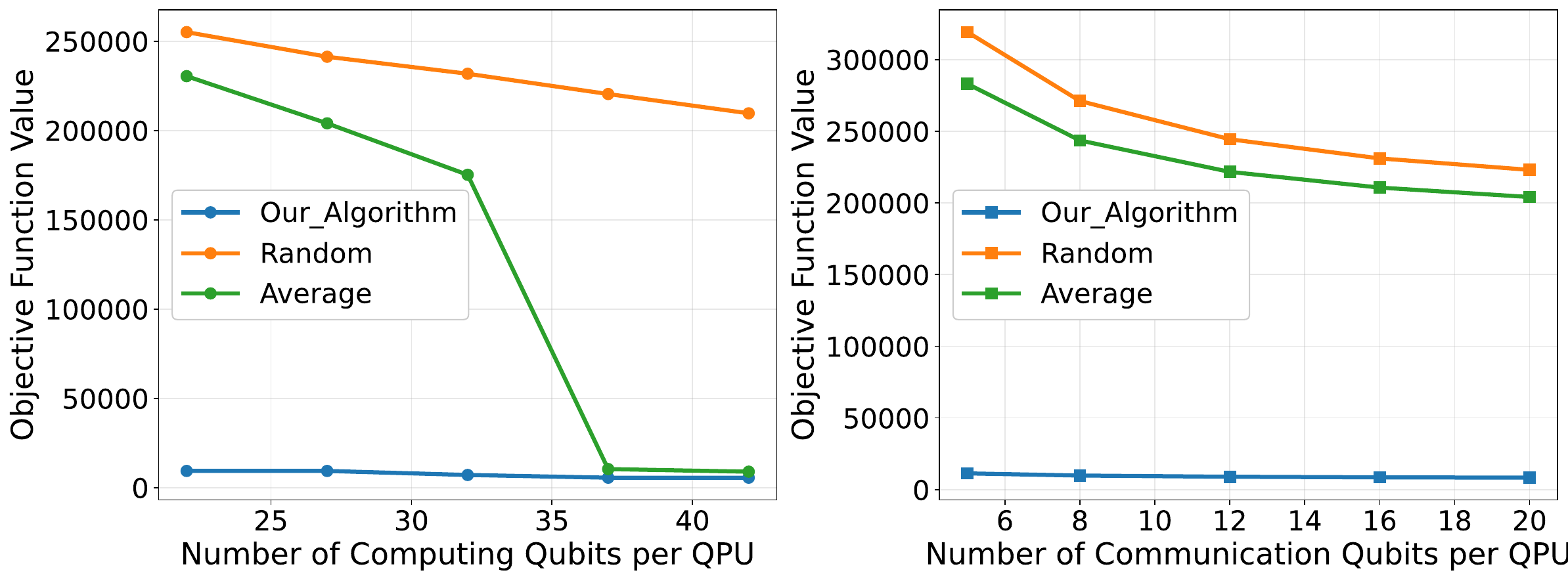}
    \caption{qugan\_n111 circuit}
    \label{fig:qugan}
\end{subfigure}

\caption{Impact of computing or communication qubit counts across four real-world quantum circuits.}
\label{fig:ep_real_cir}
% \vspace{-4mm}
\end{figure*}

\begin{figure}[t]
\centering
\includegraphics[width=0.45\textwidth]{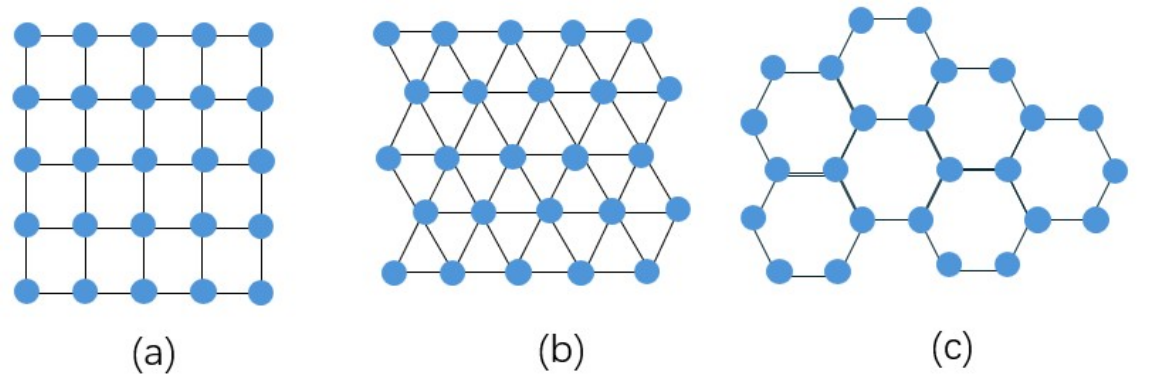}
\caption{QPU topologies. (a) Square topology, E/N=1.60. (b) Triangle topology, E/N=2.24. (c) Hexagonal topology, E/N=3. }\label{fig:qpu_topology}
\vspace{-3mm}
\end{figure}

\begin{figure}[t]
\centering
\includegraphics[width=0.45\textwidth]{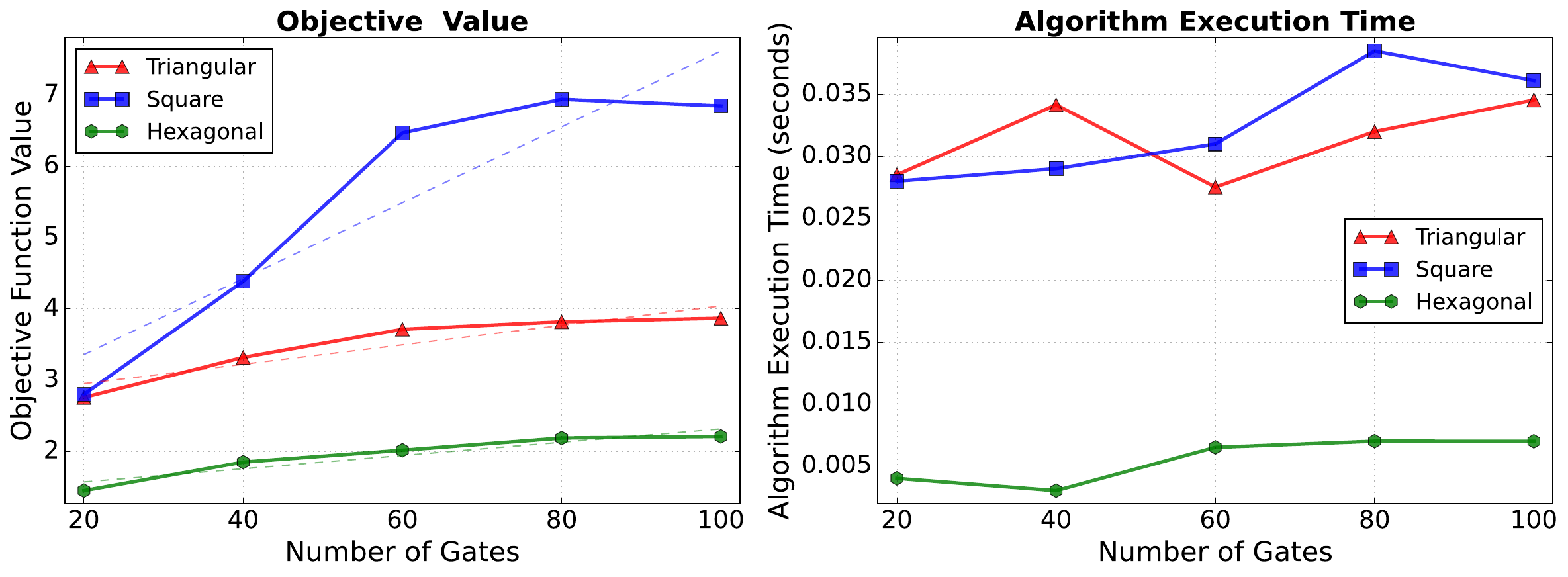}
\caption{Comparison of QPU topologies.}\label{fig:topo_compare}
\vspace{-4mm}
\end{figure}

\subsection{Evaluation Setting}
\subsubsection{Implementation.}
Since no public simulator supports distributed quantum clouds with per-slot EPR budgets, we implemented a lightweight discrete–event simulator in \texttt{Python}.  Quantum circuits are parsed via \texttt{Qiskit}~\cite{Qiskit}, and graph routines rely on \texttt{NetworkX}~\cite{hagberg2008exploring}.  The greedy constructor and the simulated–annealing improver are written in pure NumPy for reproducibility and speed.  All experiments were run on a single CPU core unless stated otherwise.

\subsubsection{Topology Settings.} Unless otherwise specified, our DQC framework is configured with 5 QPUs by default, each equipped with 20 computing qubits and 10 communication qubits. The inter-QPU topology is randomly generated, meaning that the paths between QPUs are assigned at random. The generated topology is connected, so every pair of QPUs is reachable and a remote gate between any two QPUs can be served by entanglement routing along their shortest path (Section~\ref{sec:preliminaries}); the cost $C_{u,v}$ is the corresponding hop count. The placement (first stage) minimizes the distance-weighted interaction cost $\sum_j w_{qj} C_{u,\pi(j)}$, so it keeps strongly interacting qubits on nearby QPUs and avoids costly cuts. If a pair of QPUs were unreachable, the corresponding cost would be infinite and the placement would avoid cutting across it; a remote gate that cannot be routed at all is detected and reported as infeasible, as for the capacity case.

\subsubsection{Multi-tenant Workload Construction.} To rigorously benchmark against CloudQC in multi-tenant scenarios, we construct composite workloads by merging multiple independent quantum circuits into a single non-connected DAG. This setup simulates concurrent job submissions from different users. For UNIQ, this composite DAG is treated as a single global optimization task, allowing the JIT scheduler to identify and exploit idle communication resources across independent circuits, thereby ensuring a fair, like-for-like comparison with CloudQC’s multi-tenant capabilities.

\subsubsection{Evaluation Metrics.} We evaluate system performance using four metrics: circuit runtime, algorithm execution time, the objective value defined in Eq.~(1), and EPR pairs utilization. UNIQ is designed to minimize the first three metrics while maximizing EPR pairs utilization.

\subsection{Representative Example}
We demonstrate UNIQ effectiveness through a representative example involving a quantum circuit with 50 qubits and 50 CNOT gates, each assigned to one time slot. The gates are sequentially indexed from 1 to 50 and randomly placed across the qubits. The left diagram shown in Fig.~\ref{fig:ep_example} visualizes the time slot scheduling of the gates, where red boxes denote remote CNOT gates and blue boxes denote local CNOT gates. For clarity, we provide an example explanation: gates indexed 1–5 are local gates and are all scheduled in $t_1$; gates indexed 11, 17, 21 are remote gates and are all scheduled in $t_2$. This dense packing of gates in early slots highlights UNIQ’s ability to prioritize independent operations, thereby minimizing total circuit runtime. The right diagram presents the EPR pairs utilization across time slots. A clear concentration of EPR consumption in earlier time slots reflects UNIQ’s efficiency in managing limited communication resources.

\subsection{Algorithm Comparison}
We compare UNIQ with existing algorithms across various quantum circuits and QPU topologies. The following algorithms are used as baselines. a) Simulated Annealing (SA): A heuristic method that probabilistically escapes local optima to find near-optimal solutions. b) Random: Randomly map qubits to QPUs, sample any precedence‑respecting gate order, and defer conflicts to the earliest feasible slot. c) Average: Inter-QPU communication capacity is uniformly distributed, giving each cross gate an equal share.
% \begin{itemize}
%      \item Simulated Annealing (SA): A heuristic method that probabilistically escapes local optima to find near-optimal solutions.
%      \item Random: Randomly map qubits to QPUs, sample any precedence‑respecting gate order, and defer conflicts to the earliest feasible slot.
%      \item Average: Inter-QPU communication capacity is uniformly distributed, giving each cross gate an equal share.
% \end{itemize}

\subsubsection{Different user-defined circuits.} Table~\ref{tab:user_circuit} presents the characteristics of user-defined circuits with various scales. The results are shown in Fig.~\ref{fig:ep_user_cir}, where the left plot compares the objective values, and the right plot presents the corresponding algorithm execution times. All algorithms are configured to return a feasible solution within a 600-second time limit. The results show that UNIQ performs well across the evaluated task sizes, with short execution times (about $0.01$ seconds) on these benchmarks. In contrast, the execution time of the SA algorithm is much longer (exclude it from subsequent simulations), and the objective values generated by the Random and Average algorithm are much higher.

\subsubsection{Different computing or communication qubits.} 

To evaluate the impact of computing and communication qubits, we selected four representative real-world quantum circuits: Bernstein-Vazirani (\texttt{bv\_n70}), Cat state (\texttt{cat\_n130}), Greenberger-Horne-Zeilinger (\texttt{ghz\_n127}), and Quantum GAN (\texttt{qugan\_n111}), summarized in Table~\ref{tab:real_circuit}. While we present results for these specific benchmarks, it is important to note that our UNIQ framework is designed to be generally applicable to most quantum circuits, regardless of their specific gate structures or algorithmic categories.
We vary the number of computing and communication qubits independently to evaluate their impact on performance. As shown in Fig.~\ref{fig:ep_real_cir}, increasing either type of qubit generally leads to lower objective values. The two qubit types act on the objective through different routes. The communication-qubit count $E_u$ does not appear in Eq.~(1); it enters through the per-slot EPR budget (Constraint~i): with more communication qubits, more remote gates can be served in the same slot, so gates are scheduled in earlier slots and the weighted flow-time term $\sum_t\sum_g t\,z_{g,t}=\sum_g \tau(g)$ shrinks, while the communication-cost term, which depends only on the placement, is unchanged. The computing-qubit count instead affects the placement: more computing capacity allows tighter clustering, which lowers the number and distance of remote gates and hence the communication-cost term. The y-axis represents the dimensionless weighted cost calculated via Eq. (1), integrating both latency (in time slots) and communication overhead (in hops). Across the evaluated settings, our algorithm achieves the lowest objective value among the baselines. The Random baseline uses a random placement and gate order, so its objective is a high-variance sample: its curve is not monotone, and a local rise-then-fall at a particular qubit count reflects sampling variance over the few trials, not a systematic effect. Averaged over many trials the Random objective is essentially flat in the computing-qubit count (a random placement does not exploit the added capacity), whereas UNIQ is deterministic and decreases monotonically.

\subsubsection{Unified qubit pool.}
We further examine a QPU that maintains a single pool of $K$ qubits and study how splitting it between computing capacity ($Cap=K-E$) and communication qubits ($E$) affects performance. As shown in Fig.~\ref{fig:qubit_split_ms}, increasing the communication share lowers the runtime by parallelizing EPR generation and then saturates, while the communication cost rises because less computing capacity spreads qubits over more QPUs and creates more remote gates. The combined objective therefore has an interior optimum, and its location depends on the circuit: circuits with little remote-gate concurrency favor the smallest communication budget, whereas more parallel circuits are minimized near a communication fraction of $0.17$ to $0.33$. This trade-off is why the computing and communication capacities are accounted for separately; a unified pool corresponds to optimizing the split.

\begin{figure}[t]
\centering
\includegraphics[width=0.45\textwidth]{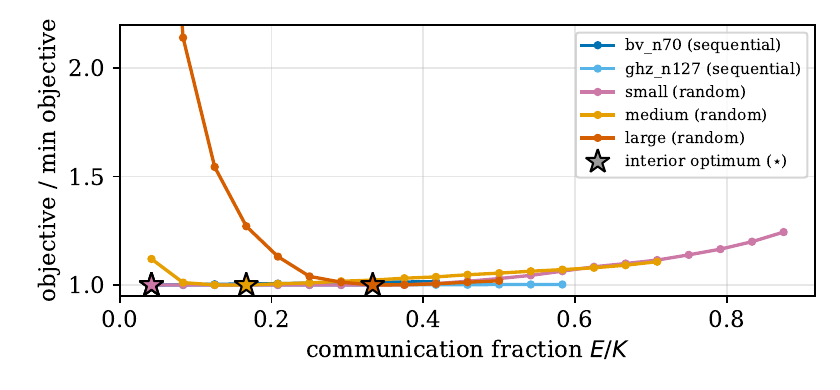}
\caption{Normalized objective versus communication fraction $E/K$ for five circuits; stars mark the interior optima, whose location tracks remote-gate concurrency.}
\label{fig:qubit_split_ms}
\end{figure}

\subsubsection{Different QPU topologies.} The three topology structures are shown in Fig~\ref{fig:qpu_topology}. Each topology contains 25 nodes (25 QPUs), and the edge to node ratio E/N (E is the number of edges and N is the number of nodes) reflects the degree of topology connectivity. As shown in Fig.~\ref{fig:topo_compare}, a higher E/N ratio leads to better system performance, as higher connectivity improves the flexibility of gate scheduling and reduces communication cost.

\subsection{Optimality Gap and Scalability}
We measure how much the staged Greedy--JIT solver gives up relative to solving the unified NIP jointly. We implemented the exact NIP with a constraint solver and compared its global optimum to Greedy--JIT on $810$ small instances of up to $10$ qubits, solved to proven optimality. As shown in Fig.~\ref{fig:optgap_ms}(a), at a representative communication budget Greedy--JIT reaches the exact optimum on $42\%$ of the instances, with a mean objective gap of $3.8\%$ (median $2.6\%$). Separating the two stages, the scheduling phase is essentially optimal (median gap $0\%$, mean $0.3\%$), in agreement with Proposition~\ref{prop:sched}, while the residual gap comes almost entirely from the greedy placement, whose communication term is about one third above optimal; the placement is therefore the looser stage and a target for future refinement. Only the adversarial minimum budget of a single communication qubit per QPU inflates the objective gap, to about $13\%$ on average. Fig.~\ref{fig:optgap_ms}(b) shows why the model is solved by decomposition: the exact solver is tractable only on small workloads, timing out beyond a few tens of gates, whereas Greedy--JIT returns a feasible schedule in under a millisecond at every size. Solving the NIP exactly is therefore unavailable at the scale of the evaluated circuits, and the decomposition is what makes the unified model usable in practice.

\begin{figure}[t]
\centering
\includegraphics[width=0.48\textwidth]{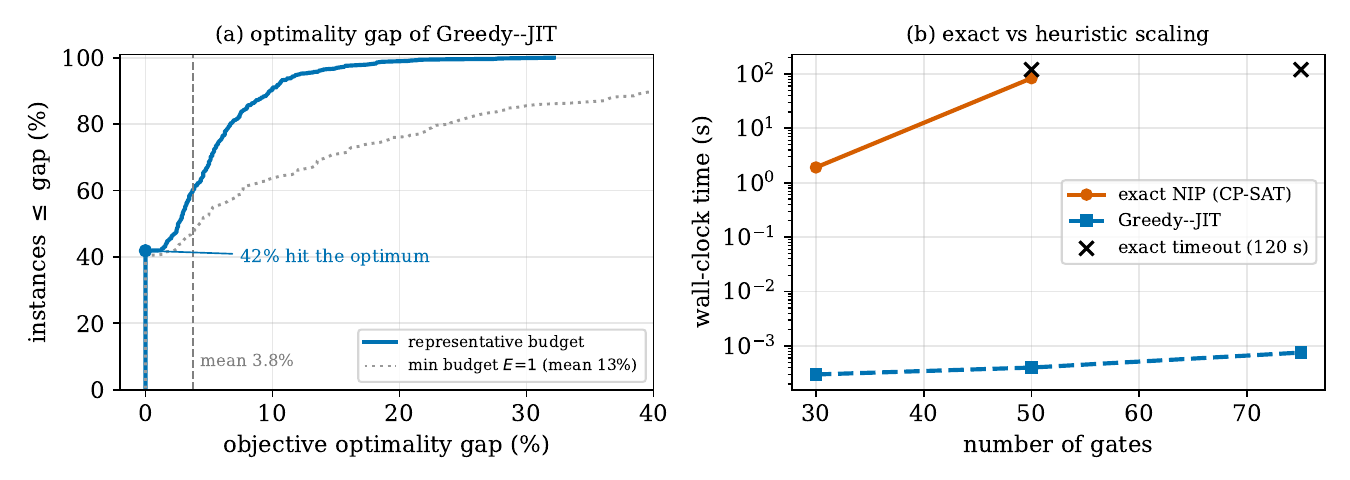}
\caption{Greedy--JIT versus the exact NIP. (a) Optimality-gap CDF ($42\%$ exact, mean $3.8\%$; dotted: single-qubit budget). (b) Runtime versus size: the exact solver times out beyond tens of gates ($\times$); Greedy--JIT stays near a millisecond.}
\label{fig:optgap_ms}
\end{figure}

\subsection{EPR Coherence and Fidelity}
The coherence constraint~(j) bounds EPR storage time and keeps the worst-case stored fidelity above the target on every platform, without changing any reported runtime, since all runtimes use the in-slot default. We evaluate how the storage time of EPR pairs affects their fidelity under the model of Constraint~(j). We reimplemented the Greedy--JIT constructor and measured the storage time of every EPR over 30 random topologies on seven circuits, under three generation policies that share the same execution slots: in-slot (default), unconstrained pre-establishment, and pre-establishment under Constraint~(j). Each pair's fidelity follows $F(\Delta)=\tfrac12+(F_0-\tfrac12)e^{-\Delta t_{ep}/T_{coh}}$ with $F_0=0.98$, evaluated at four platform coherence times.

\begin{figure}[t]
\centering
\includegraphics[width=0.48\textwidth]{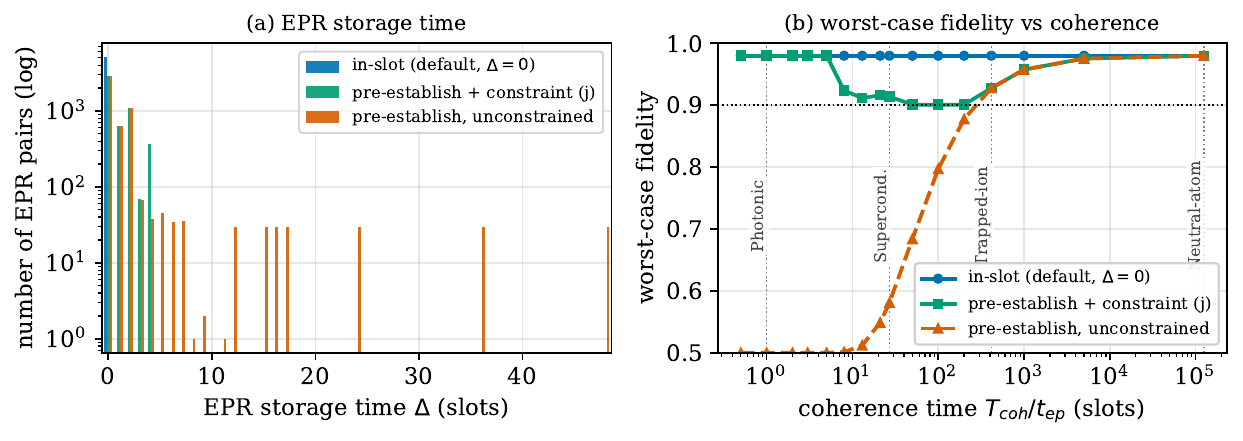}
\caption{EPR storage and fidelity. (a) Constraint~(j) caps the unconstrained model's long storage tail. (b) Worst-case fidelity: unconstrained storage collapses toward $F=0.5$ on short-coherence hardware; the default and Constraint~(j) stay above $F=0.9$ everywhere.}
\label{fig:coherence_ms}
\end{figure}

As shown in Fig.~\ref{fig:coherence_ms}, under the default in-slot rule every EPR is consumed in its generation slot ($\Delta=0$), so its stored fidelity equals $F_0$ on every platform. When pre-establishment is enabled, the unconstrained model permits storage times up to tens of slots, which drives the stored-pair fidelity toward $0.5$ on short-coherence hardware (for example a photonic link). Constraint~(j) caps the storage time at the budget $L_{coh}$ and keeps the worst-case fidelity at or above $0.914$ for a target $F_{th}=0.9$ on all four platforms, while leaving every reported runtime unchanged, since all runtimes use the in-slot default. The benefit of pre-establishment therefore varies with the platform through $L_{coh}$: on a photonic link $L_{coh}=0$ forces the in-slot default, whereas on trapped-ion and neutral-atom hardware $L_{coh}\ge 75$ leaves pre-establishment essentially unconstrained; in all cases the in-slot default keeps every reported result independent of the platform coherence time.

\subsection{Framework Comparison}
\begin{figure}[t]
\centering
\includegraphics[width=0.45\textwidth]{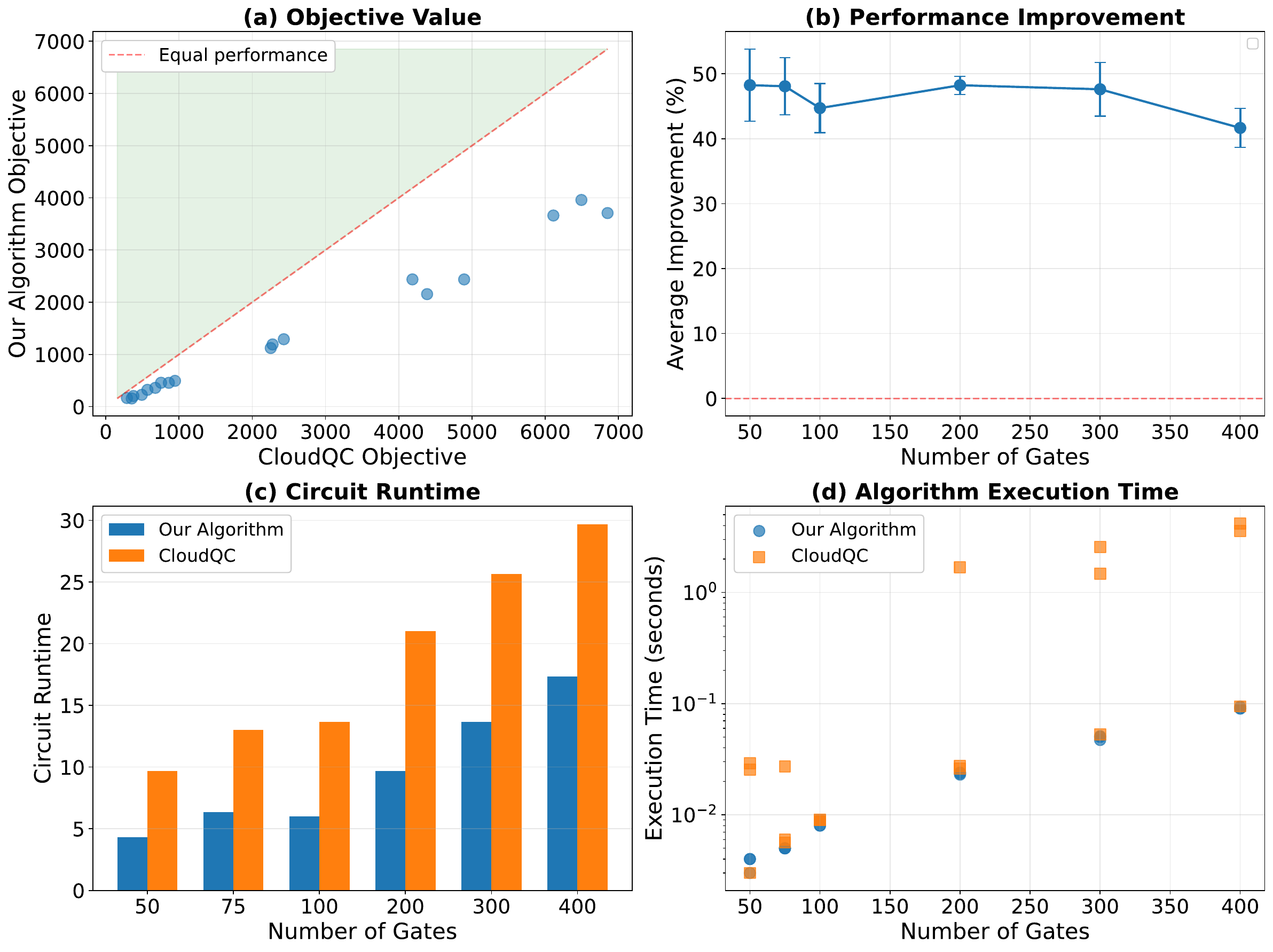}
\caption{Comparison with CloudQC.}\label{fig:cloudqc}
\end{figure}
Under matched settings, UNIQ reduces both the objective and the circuit runtime relative to CloudQC. We compare UNIQ with existing DQC frameworks, specifically CloudQC~\cite{zhou2025cloudqc}, which pioneered the progressive optimization approach for multi-tenant DQC. While CloudQC optimizes allocation and scheduling sequentially, UNIQ employs a unified approach via the NIP model and JIT scheduling. To ensure a fair comparison, the baseline and our framework are evaluated under similar problem settings. Specifically, they should cover all three stages of DQC and adopt the Cat-Comm protocol for remote gate execution. The two methods are run under matched settings: identical benchmark circuits and multi-tenant composite workloads, the same cloud configuration and the Cat-Comm latencies of Table~\ref{tab:epr_latency}, and the same three metrics (the objective of Eq.~(\ref{eq:objective}), circuit runtime, and algorithm execution time); they differ only in the optimization, which is sequential and on-demand for CloudQC and unified with parallel EPR generation for UNIQ. Inspired by CloudQC~\cite{zhou2025cloudqc}, UNIQ is also suitable for multi-tenant scenarios. Specifically, our method can be viewed as globally sorting all CNOT gates across multi-tenants. Based on this similarity, we select CloudQC as the baseline framework for comparison.

We conducted three simulations for each circuit scale. As shown in Fig.~\ref{fig:cloudqc}(a), x- and y-coordinates of each blue point represent the average objective values of CloudQC and UNIQ across different circuit scales. For instance, the last blue point indicates that CloudQC's objective is close to 7000, while UNIQ's is around 4000. Figure~\ref{fig:cloudqc}(b) shows that UNIQ reduces the objective value of Eq.~(\ref{eq:objective}) by nearly 50\% compared to CloudQC. Because CloudQC optimizes a different internal objective, we also report the circuit runtime (Fig.~\ref{fig:cloudqc}(c)), a method-neutral physical metric, on which UNIQ is likewise lower. UNIQ achieves lower runtime because its JIT scheduling identifies fragmented idle time slots in the multi-tenant workload and uses them to parallelize EPR generation, increasing throughput on the same hardware.

To identify where this advantage originates, we ablate the two components that distinguish UNIQ from a sequential, on-demand approach: the JIT parallel EPR generation and the distance-aware joint placement. On circuits with concurrent remote gates, parallel EPR generation accounts for most of the objective reduction (over $90\%$ on the larger circuits), because it shrinks the flow-time term, while the placement contributes the remainder and dominates the communication-cost component; on sequential circuits, where there is nothing to parallelize, the two coincide. A direct numerical comparison with the recent joint-optimization works of Liu et al.~\cite{liu2025co} and Bahrani et al.~\cite{bahrani2025resource} is not like-for-like, because they solve different decision problems (network-topology co-design and batch whole-circuit scheduling) under different assumptions, and their exact-solver formulations target different problem sizes, which makes a direct large-scale numerical comparison impractical under our benchmark setting (Section~\ref{sec:evaluation}, optimality-gap study). We therefore compare qualitatively and use CloudQC, which addresses the same problem setting, as the closest matched experimental baseline.
Figure~\ref{fig:cloudqc}(c) presents the average circuit runtime of CloudQC and UNIQ, while each point in Fig.~\ref{fig:cloudqc}(d) shows the comparison of algorithm execution time under each simulations. Across the evaluated settings, UNIQ outperforms CloudQC in both circuit runtime and algorithm execution time.

\subsection{Robustness and Sensitivity}
We assess robustness along two axes that the slot model abstracts. First, we make local-gate latency explicit: the physical runtime is $\sum_t \max(t_{ep}, L(t)\,t_{local})$, where $L(t)$ is the number of local gates in slot $t$ and $\rho=t_{local}/t_{ep}$ ($\rho\!\to\!0$ recovers the slot-model makespan). As shown in Fig.~\ref{fig:robustness}(a), UNIQ stays ahead of a naive baseline (count-only placement, serial EPR) across the full range of $\rho$, with the objective ratio rising only from $0.62$ to $0.78$; the advantage is largest in the EPR-dominated regime of current hardware ($\rho\approx0.08$) and never disappears, because the local-gate term is method-independent. The slot model's absorption of local gates is exact while a slot's local gates fit within one $t_{ep}$ (about $t_{ep}/t_{2q}\approx12$ two-qubit gates). Second, under heterogeneous per-QPU computing capacities and communication budgets, Fig.~\ref{fig:robustness}(b) shows UNIQ remains feasible on $99$--$100\%$ of instances and keeps its mean objective $3$--$4\%$ below a deterministic count-only placement and $30$--$39\%$ below the mean of random placements at every heterogeneity level, the margin widening as heterogeneity grows. As for parameter sensitivity, the communication budget and the computing/communication split behave as in the studies above (runtime saturates once the budget meets the peak per-slot demand; the split has an interior optimum), and the objective weights $\alpha,\beta$ leave the solution unchanged, since the placement and the scheduler greedily minimize the two objective terms independently.

\begin{figure}[t]
\centering
\includegraphics[width=0.48\textwidth]{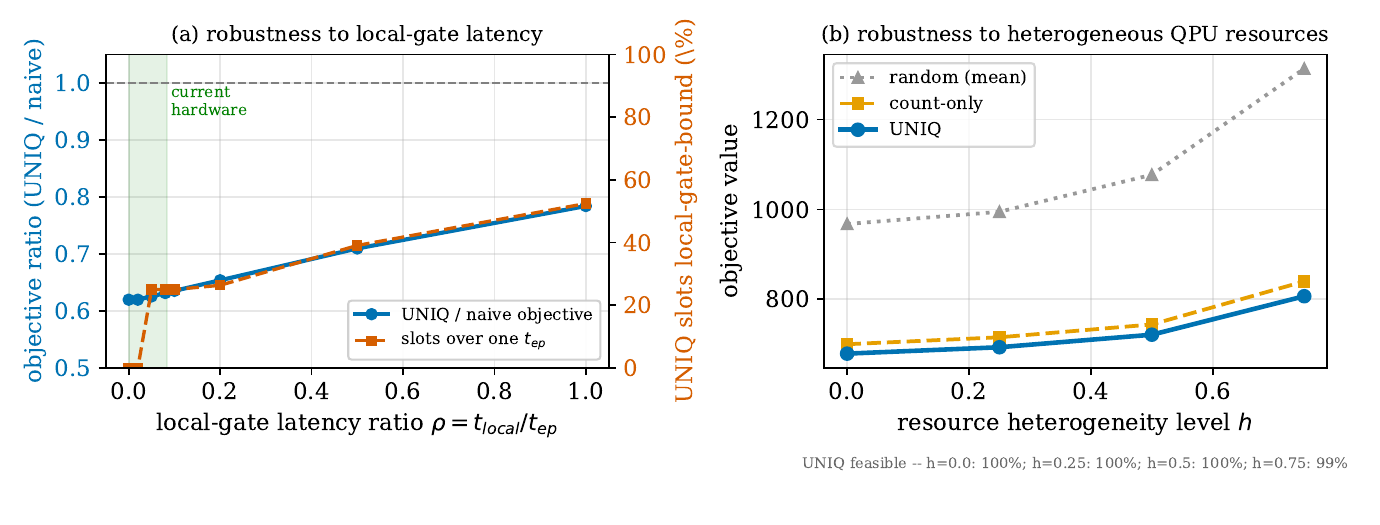}
\caption{Robustness. (a) UNIQ stays ahead of the naive baseline across the local-gate ratio $\rho$ (right axis: slots exceeding one $t_{ep}$). (b) Under heterogeneous resources, UNIQ stays below a random placement and remains feasible.}
\label{fig:robustness}
\end{figure}

\section{Conclusion\label{sec:conclusion}}
We introduced UNIQ, a unified optimization approach for the scalability and communication bottlenecks in DQC. Existing approaches optimize qubit allocation, entanglement management, and gate scheduling separately; UNIQ couples these interdependent stages in one NIP model. Because solving the NIP exactly does not scale, we solve it with a two-stage constructor that combines greedy spatial mapping with JIT temporal scheduling, and we measure its small optimality gap against the exact model on instances small enough to solve exactly. UNIQ generates the EPR pairs for concurrent remote gates in parallel on idle communication qubits, masking the latency of remote CNOT gates. Across diverse quantum circuits and QPU topologies, UNIQ lowers both circuit runtime and communication cost over existing algorithms and systems.

\section*{Acknowledgment}
The work of X. Zhang was supported in part by the U.S. National Science Foundation under Grant NSF-2329027.

\bibliographystyle{IEEEtran}
\bibliography{ref}

\end{document}

%% file: epr_time.tex
\begin{table}[t]
\centering
\caption{Comparison of EPR pair generation latency.}
\label{tab:epr_latency}
\small
\begin{tabular}{|l|l|l|}
\hline
\textbf{Operation}         & \textbf{Variable Name} & \textbf{Latency}      \\ \hline
Single-qubit gate          & $t_{1q}$               & $\sim$0.1 CX          \\ \hline
CX and CZ gate             & $t_{2q}$               & 1 CX                  \\ \hline
Measurement                & $t_{ms}$               & 5 CX                  \\ \hline
EPR preparation            & $t_{ep}$               & $\sim$12 CX           \\ \hline
\end{tabular}
\vspace{-3mm}
\end{table}

%% file: notation.tex
\begin{table}[t]
\centering
\caption{Summary of symbols and definitions.}
\footnotesize
\begin{tabularx}{0.48\textwidth}{lX}
\toprule
\textbf{Parameters} & \textbf{Description} \\
\midrule
$Q$ & logical qubits. \\
$\mathcal{G}$ & CNOT gate. $g=(i_g, j_g)$ denote the two logical qubits of CNOT gate $g$.\\
$\mathcal{P}$ & precedence relations. $(g' \rightarrow g)$ indicates that the gate $g'$ must be executed before the gate $g$.\\
$\mathcal{U}$ & QPU nodes. \\
$\mathcal{E} = \{E_u, u \in \mathcal{U}\}$ & number of communication qubits on QPU $u$. \\
$C_{uv}$ & communication cost on link $(u, v)$. We set $C_{uu}=0$. \\
$Cap_{u}$ & The maximum number of qubits in QPU $u$ (QPU capacity). \\
$H = |\mathcal{G}|$ & time-slot upper bound. \\
$\alpha$ & weight coefficients in the objective function.\\
\midrule
\textbf{Decision Variables} & \textbf{Description} \\
\midrule
$\pi_{q,u} \in \{0,1\}$ & assign qubit $q$ to QPU $u$. \\
$z_{g,t} \in \{0,1\}$ & gate $g$ is executed in slot $t$. \\
$y_{g,t} \in \{0,1\}$ & EPR pair for gate $g$ is built in slot $t$. \\
\midrule
\textbf{Auxiliary Variables} & \textbf{Description} \\
\midrule
$s_{u,v,t} \in \mathbb{Z}_{\ge 0}$ & 
\makecell[tl]{EPR stored on QPUs $(u,v)$ after time \\slot $t$.} \\
$\delta_{g} \in \{0,1\}$ & 
\makecell[tl]{$\delta_g = 1$ if logical qubits $i_g$, $j_g$ are on \\different QPUs; 0 otherwise.} \\
\bottomrule
\end{tabularx}
\label{tab:parameters}
\vspace{-3mm}
\end{table}

%% file: user_circuit.tex
\begin{table}[t]
\centering
\caption{Characteristics of user-defined circuits.}
\small
\begin{tabular}{|c|c|c|c|c|c|}
\hline
Circuits & Gates & Qubits & QPU & Comput\_Qubit & Commu\_Qubit \\ \hline
tiny     & 10    & 10     & 2   & 5              & 10           \\ \hline
small    & 20    & 15     & 3   & 5              & 10           \\ \hline
medium   & 50    & 32     & 4   & 8              & 10           \\ \hline
large    & 100   & 60     & 5   & 12             & 10           \\ \hline
\end{tabular}
\label{tab:user_circuit}
\vspace{-3mm}
\end{table}

%% file: circuit_info.tex
\begin{table}[t]
\centering
\caption{Characteristics of real-world quantum circuits. The qubit counts are the evaluated instance sizes.}
\small
\begin{tabular}{|c|c|c|c|}
\hline
Circuits    & Qubits & 2 Qubit Gates & Circuit Depth \\ \hline
bv\_n70     & 70     & 36            & 40            \\ \hline
cat\_130    & 20     & 15            & 3             \\ \hline
ghz\_n127   & 50     & 32            & 4             \\ \hline
qugan\_n111 & 100    & 60            & 5             \\ \hline
\end{tabular}
\label{tab:real_circuit}
\vspace{-3mm}
\end{table}